\documentclass[11pt]{article}

\usepackage[margin=1in]{geometry}

\usepackage{fullpage}
\usepackage{amsthm,amssymb,amsmath}  
\usepackage{xspace,enumerate}
\usepackage[utf8]{inputenc}
\usepackage{thmtools}
\usepackage{thm-restate}
\usepackage{authblk}

  \theoremstyle{plain}
  \newtheorem{theorem}{Theorem}
  \newtheorem{lemma}{Lemma}  
  \newtheorem{corollary}{Corollary}

  \newtheorem{observation}{Observation}
  \theoremstyle{definition}
  \newtheorem{definition}{Definition}
  \newtheorem{example}{Example}
  \newtheorem{remark}[definition]{Remark}
 
  \newtheorem*{claim}{Claim}
 
\title{Longest Common Substring Made Fully Dynamic}

\author[1]{Amihood Amir}
\author[2,3]{Panagiotis Charalampopoulos\thanks{Partially supported by Israeli Science Foundation (ISF) grant 794/13.}}
\author[2]{Solon P. Pissis}
\author[4]{Jakub Radoszewski}

\affil[1]{
Department of Computer Science, Bar-Ilan University, Ramat Gan, Israel\\
\texttt{amir@esc.biu.ac.il}
}

\affil[2]{
Department of Informatics, King's College London, London, UK\\
\texttt{[panagiotis.charalampopoulos,solon.pissis]@kcl.ac.uk}
}

\affil[3]{
Efi Arazi School of Computer Science, The Interdisciplinary Center Herzliya, Herzliya, Israel\\
}

\affil[4]{
Institute of Informatics, University of Warsaw, Warsaw, Poland\\
\texttt{jrad@mimuw.edu.pl}
}

\date{\vspace{-5ex}}

  \usepackage{hyperref}
  \hypersetup{colorlinks=true,citecolor=blue}

  \usepackage{xspace}
  \usepackage{makecell}
  \usepackage[ruled,noline,noend]{algorithm2e}
  \usepackage{tabularx}
  \usepackage{comment}
  \usepackage{todonotes}
  \usepackage{thmtools}
  \usepackage{thm-restate}
  \usepackage[capitalise]{cleveref}
  \usepackage{verbatim}
  
  \usepackage{tikz}
  \usepackage{tikz-qtree}
	\usetikzlibrary{trees}
	\usetikzlibrary{decorations.pathmorphing}
	\usetikzlibrary{decorations.markings}
	\usetikzlibrary{decorations.pathmorphing,shapes}
	\usetikzlibrary{calc,decorations.pathmorphing,shapes}
	\usetikzlibrary{snakes}

\def\dd{\mathinner{.\,.}}
\newcommand{\cO}{\mathcal{O}}
\newcommand{\Oh}{\cO}
\newcommand{\tr}{\mathcal{T}}

\newcommand{\SA}{\textsf{SA}}

\newcommand{\LCE}{\textsf{LCE}}
\newcommand{\LCS}{\textsf{LCS}}
\newcommand{\lcp}{\textsf{lcp}}
\newcommand{\lcpstring}{\textsf{lcpstring}}
\newcommand{\lcsstring}{\textsf{lcsstring}}
\newcommand{\lcs}{\textsf{lcs}}

\newcommand{\RMQ}{\textsf{RMQ}}
\newcommand{\range}{\textit{range}}
\newcommand{\Ohtilde}{\tilde{\cO}}
\newcommand{\id}{\textit{id}}
\newcommand{\nnext}{\textit{next}}
\newcommand{\ccount}{\textit{count}}
\newcommand{\iinsert}{\textit{insert}}

\newcommand{\minPref}{\mathit{minPref}}
\newcommand{\maxPref}{\mathit{maxPref}}

\newcommand{\LSPal}{\textit{LSPal}\xspace}
\newcommand{\M}{\mathcal{M}}
\newcommand{\B}{\mathcal{B}}

\renewcommand{\P}{\mathcal{P}}
\newcommand{\Q}{\mathcal{Q}}
\newcommand{\T}{\mathcal{T}}
\newcommand{\F}{\mathcal{F}}
\newcommand{\nn}{\mathit{succ}}
\newcommand{\pp}{\mathit{pred}}
\newcommand{\problem}{\textsc{Two String Families LCP}\xspace}

\newcommand{\maxPairLCP}{\mathrm{maxPairLCP}}

\renewcommand{\S}{\mathbf{S}}

 \newcommand{\defproblem}[3]{
  \vspace{2mm}
\noindent\fbox{
  \begin{minipage}{0.96\textwidth}
  #1\\
  {\bf{Input:}} #2  \\
  {\bf{Output:}} #3
  \end{minipage}
  }
  \vspace{2mm}
}

 \newcommand{\defDSproblem}[3]{
  \vspace{2mm}
\noindent\fbox{
  \begin{minipage}{0.96\textwidth}
  #1\\
  {\bf{Input:}} #2  \\
  {\bf{Query:}} #3
  \end{minipage}
  }
  \vspace{2mm}
}

\begin{document}

\maketitle


\begin{abstract}
In the longest common substring (LCS) problem, we are given two strings $S$ and $T$, each of length at most $n$, and we are asked to find a longest string occurring as a fragment of both $S$ and $T$. This is a classical and well-studied problem in computer science with a known $\Oh(n)$-time solution. In the fully dynamic version of the problem, edit operations are allowed in either of the two strings, and we are asked to report an LCS after each such operation. We present the first solution to this problem that requires sublinear time per edit operation. In particular, we show how to return an LCS in $\tilde{\cO}(n^{2/3})$ \footnote{The $\tilde{\cO}(\cdot)$ notation suppresses $\log^{\cO(1)} n$ factors.} time (or $\tilde{\cO}(\sqrt{n})$ time if edits are allowed in only one of the two strings) after each operation using $\tilde{\cO}(n)$ space.

This line of research was recently initiated by the authors [SPIRE 2017] in a somewhat restricted dynamic variant. An $\tilde{\cO}(n)$-sized data structure that returns an LCS of the two strings after a single edit operation (that is reverted afterwards) in $\tilde{\cO}(1)$ time was presented. At CPM 2018, three papers studied analogously restricted dynamic variants of problems on strings.
We show that our techniques can be used to obtain fully dynamic algorithms for several classical problems on strings, namely, computing the longest repeat, the longest palindrome and the longest Lyndon substring of a string. The only previously known sublinear-time dynamic algorithms for problems on strings were obtained for maintaining a dynamic collection of strings for comparison queries and for pattern matching with the most recent advances made by Gawrychowski et al.\ [SODA 2018] and by Clifford et al.\ [STACS 2018].

As an intermediate problem we consider computing the required result on a string with a given set of $k$ edits, which leads us, in particular, to answering \emph{internal} queries on a string. The input to such a query is specified by a substring (or substrings) of a given string. Data structures for answering internal queries that were proposed by Kociumaka et al.\ [SODA 2015] and by Gagie et al.\ [CCCG 2013] are used, along with new data structures for several types of queries that we develop --- these queries include computing an LCS on multiple substrings. However, as it is likely that the most general internal LCS queries cannot be answered efficiently,  we also design a different scheme that requires an application of difference covers. Other ingredients of our algorithms include the suffix tree, heavy-path decompositions, orthogonal range queries, and string periodicity.

\end{abstract}
\thispagestyle{empty}
\clearpage
\setcounter{page}{1}

\section{Introduction}

\subsection{Longest Common Substring (LCS) Problem}
In the longest common substring (LCS) problem, also known as longest common factor problem, we are given two strings $S$ and $T$, each of length at most $n$, and we are asked to find a longest string occurring in both $S$ and $T$. This is a classical and well-studied problem in theoretical computer science arising out of different practical scenarios. In particular, the LCS is a widely used string similarity measure. Assuming that $S$ and $T$ are over a linearly-sortable alphabet, the LCS problem can be solved in $\cO(n)$ time and space~\cite{Gusfield:1997:AST:262228,Chi1992}. This is clearly time-optimal taking into account that $\Omega(n)$ time is required to read $S$ and $T$ for large integer alphabets. A series of studies has thus been dedicated to improving the working space~\cite{DBLP:conf/esa/KociumakaSV14,DBLP:conf/cpm/StarikovskayaV13}; that is, the {\em additional} space required for the computations, not taking into account the space required to store $S$ and $T$. 

As it is quite common to account for potential alterations within textual data (e.g.~DNA sequences), it is natural to define the LCS problem under a distance metric model. The problem is then to find a longest substring of $S$ that is at distance at most $k$ from any substring of $T$. It has received much attention recently, in particular due to its applications in computational molecular biology~\cite{DBLP:journals/jcb/UlitskyBTC06}. Under the Hamming distance model (substitution operations), it is known as the LCS problem with $k$ mismatches. We refer the interested reader to~\cite{Abboud:2015:MAP:2722129.2722146,starikovskaya:LIPIcs:2016:6072,DBLP:journals/jcb/ThankachanAA16,DBLP:conf/cpm/Charalampopoulos18} and references therein; see also~\cite{DBLP:conf/recomb/ThankachanACA18} for the edit distance model.

In~\cite{starikovskaya:LIPIcs:2016:6072}, Starikovskaya mentions that an answer to the LCS problem ``is not robust and can vary significantly when the input strings are changed even by one character'', posing implicitly the following natural question: Can we compute a new LCS after editing $S$ or $T$ in $o(n)$ time? To this end Amir et al.\ in~\cite{Amir2017} presented a solution for the restricted case, where any {\em single} edit operation is allowed. We call this problem \textsc{LCS after One Edit}. Specifically, the authors presented an $\tilde{\cO}(n)$-space data structure that can be constructed in $\tilde{\cO}(n)$ time supporting $\tilde{\cO}(1)$-time computation of an LCS of $S$ and $T$, after one edit operation (that is reverted after the query) is applied on $S$.
This work initiated a new line of research on analogously restricted dynamic variants of string problems~\cite{DBLP:conf/cpm/FunakoshiNIBT18,DBLP:conf/cpm/UrabeNIBT18}. 
%
Moreover, Abedin et al.~\cite{DBLP:conf/cpm/AbedinH0T18} improved the complexities of the data structure proposed by Amir et al.~\cite{Amir2017} by shaving some $\log^{\Oh(1)} n$ factors, whereas other restricted variants of the dynamic LCS problem have been considered by Amir and Boneh in~\cite{amir_et_al:LIPIcs:2018:8698}. One variant was to consider substitution operations in which a character is replaced by a character not from the alphabet and these operations are only allowed in one of the strings. Another variant was to consider substitutions in one of the strings but parameterize the time complexity by the period of the static string, which is in $\Theta(n)$.

In this paper we continue this line of research and show a solution for the {\em general} version of the problem, namely, the fully dynamic case of the LCS problem. Given two strings $S$ and $T$, we are to answer the following type of queries in an on-line manner: perform an edit operation (substitution, insertion, or deletion) on $S$ or on $T$ and then return an LCS of the new $S$ and $T$. We allow preprocessing the initial $S$ and $T$. We call this the {\em fully dynamic LCS} problem. Let us illustrate an example of this problem.

\begin{example}
The length of an LCS of the first pair of strings, from left to right, is {\em doubled} when one substitution operation, $S[4]:=\texttt{a}$, is performed. The next substitution, namely $T[3]:=\texttt{b}$, halves the length of an LCS.

\begin{minipage}[t]{0.18\textwidth}
$$S=\texttt{c\underline{aab}aaa}$$
$$T=\texttt{aaaa\underline{aab}}$$
\end{minipage}
\begin{minipage}[t]{0.18\textwidth}
$$S[4]:=\texttt{a}$$
\end{minipage}
\begin{minipage}[t]{0.18\textwidth}
$$S=\texttt{c\underline{aaaaaa}}$$
$$T=\texttt{\underline{aaaaaa}b}$$
\end{minipage}
\begin{minipage}[t]{0.18\textwidth}
$$T[3]:=\texttt{b}$$
\end{minipage}
\begin{minipage}[t]{0.18\textwidth}
$$S=\texttt{c\underline{aaa}aaa}$$
$$T=\texttt{aab\underline{aaa}b}$$
\end{minipage}

\end{example}

\subsection{Dynamic Problems on Strings}
Below we mention known related results on dynamic problems on strings.
\paragraph{Dynamic Pattern Matching} Finding all {\em occ} occurrences of a pattern of length $m$ in a {\em static}
text can be done in the optimal $\Oh(m+\textit{occ})$ time using suffix trees, which can be constructed in linear time~\cite{F97}. In the fully dynamic setting of this problem, we are asked to compute the new set of occurrences when
allowing for edit operations anywhere on the text. A considerable amount of work has been done on this problem~\cite{Gu:1994:EAD:314464.314675,DBLP:journals/jal/Ferragina97,DBLP:journals/siamcomp/FerraginaG98}.
The first data structure with poly-logarithmic update time and time-optimal queries was shown by Sahinalp and Vishkin~\cite{DBLP:conf/focs/SahinalpV96}.
The update time was later improved by Alstrup et al.~\cite{Alstrup:2000:PMD:338219.338645} at the expense of slightly suboptimal query time. The state of the art is the data structure by Gawrychowski et al.~\cite{DBLP:journals/corr/GawrychowskiKKL15} supporting time-optimal queries with $\Oh(\log^2 n)$ time for updates. Clifford et al.~\cite{DBLP:conf/stacs/CliffordGLS18} have recently shown upper and lower bounds for variants of exact matching with wildcards, inner product, and Hamming distance.

\paragraph{Dynamic String Collection with Comparison} We are to maintain a dynamic collection $\mathcal{W}$ of strings of total length $n$ that supports the following operations:
\begin{itemize}
\item \textsf{makestring}$(W)$: insert a non-empty string $W$;
\item \textsf{concat}$(W_1,W_2)$: insert $W_1 W_2$ to $\mathcal{W}$, for $W_1,W_2 \in \mathcal{W}$;
\item \textsf{split}$(W,i)$: split the string $W$ at position $i$ and insert both resulting strings to $\mathcal{W}$, for $W \in \mathcal{W}$;
\item \lcp$(W_1,W_2)$: return the length of the longest common prefix
of $W_1$ and $W_2$, for $W_1,W_2 \in \mathcal{W}$.
\end{itemize}
This line of research was initiated by Sundar and Tarjan~\cite{DBLP:journals/siamcomp/SundarT94}. Data structures supporting updates in polylogarithmic time were presented by Mehlhorn et al.~\cite{DBLP:journals/algorithmica/MehlhornSU97} (with a restricted set of queries) and Alstrup et al.~\cite{Alstrup:2000:PMD:338219.338645}. Finally, Gawrychowski et al.~\cite{DBLP:conf/soda/GawrychowskiKKL18} proposed a solution with the operations requiring time $\cO(\log n +|w|)$, $\cO(\log n)$, $\cO(\log n)$, and $\cO(1)$, respectively that they show to be optimal. Their collection of strings is persistent. 

\paragraph{Longest Palindrome Substring After One Edit} Palindromes (also known as symmetric strings) are one of the fundamental concepts on strings with applications in computational biology (see, e.g.,~\cite{Gusfield:1997:AST:262228}). A recent progress in this area was the design of an $\Oh(n \log n)$-time algorithm for partitioning a string into the minimum number of palindromes~\cite{DBLP:journals/jda/FiciGKK14,DBLP:conf/cpm/ISIBT14} (that was improved to $\Oh(n)$ time \cite{DBLP:conf/cpm/BorozdinKRS17} afterwards). The main combinatorial insight of these results is that the set of lengths of suffix palindromes of a string can be represented as a logarithmic number of arithmetic progressions, each of which consists of palindromes with the same shortest period. Funakoshi et al.~\cite{DBLP:conf/cpm/FunakoshiNIBT18} use this fact to present a data structure for computing a longest palindrome substring of a string after a single edit operation. This problem is called \textsc{Longest Palindrome Substring after One Edit}. They obtain $\Oh(\log \log n)$-time queries with a data structure of $\Oh(n)$ size that can be constructed in $\Oh(n)$ time. We present a fully dynamic algorithm with $\Ohtilde(\sqrt{n})$-time queries for this problem.

\paragraph{Longest Lyndon Substring After One Edit} A \emph{Lyndon string} is a string that is smaller (in the lexicographical order) than all its suffixes \cite{Lyndon}. E.g., \texttt{aabab} and \texttt{a} are Lyndon strings, whereas \texttt{abaab} and \texttt{abab} are not. Lyndon strings are an object of interest in combinatorics on words especially due to the Lyndon factorization theorem \cite{ChenFoxLyndon} that asserts that every string can be uniquely decomposed into a non-decreasing sequence of Lyndon strings. Recently Lyndon strings have found important applications in algorithm design~\cite{DBLP:conf/soda/Mucha13} and were used to settle a known conjecture on the number of repetitions in a string~\cite{DBLP:conf/soda/BannaiIINTT15,DBLP:journals/siamcomp/BannaiIINTT17}.

Urabe et al.~\cite{DBLP:conf/cpm/UrabeNIBT18} presented a data structure for computing a longest substring of a string being a Lyndon string in the restricted dynamic setting of a single edit that is reverted afterwards. This problem is called Longest \textsc{Lyndon Substring after One Edit}. Their data structure can be constructed in $\Oh(n)$ time and space and answers queries in $\Oh(\log n)$ time. A simple observation of~\cite{DBLP:conf/cpm/UrabeNIBT18} is that the longest Lyndon substring of a string is always one of the factors of the Lyndon factorization. Thus this work indirectly addresses the question of maintaining the Lyndon factorization of a dynamic string. We present an algorithm that maintains a representation of the Lyndon factorization of a string with $\Ohtilde(\sqrt{n})$-time queries in the fully dynamic setting.

\subsection{Our Results}
We give the first fully dynamic algorithm for the LCS problem that works in {\em strongly sublinear} time per edit operation in any of the two strings. Specifically, for two strings of length up to $n$ our algorithm uses $\Ohtilde(n)$ space and computes the LCS of two strings after each subsequent edit operation in $\Ohtilde(n^{2/3})$ time. In the special case that edit operations are allowed only in one of the strings, $\Ohtilde(\sqrt{n})$ time is achieved. In order to ease the comprehension of the general fully dynamic LCS, we first show a solution of an auxiliary problem called \textsc{LCS after One Substitution per String} where a single substitution is allowed in both strings that is reverted afterwards, in $\Ohtilde(n)$ preprocessing time and $\Ohtilde(1)$-time queries.

The significance of this result is additionally highlighted by the following argument.
It is known that finding an LCS when the strings have wildcard characters~\cite{DBLP:conf/icalp/AbboudWW14} or when $k=\Omega(\log n)$ substitutions are allowed~\cite{DBLP:journals/corr/abs-1712-08573} in truly subquadratic time would refute the Strong Exponential Time Hypothesis (SETH)~\cite{Impagliazzo:2001:PSE:569473.569474,DBLP:journals/jcss/ImpagliazzoP01} (on the other hand, pattern matching with mismatches can be solved in $\Ohtilde(n\sqrt{n})$ time~\cite{DBLP:journals/siamcomp/Abrahamson87}). It is therefore unlikely that a fully dynamic algorithm with strongly sublinear-time queries exists for these problems: such an algorithm could be trivially applied as a black box to solve the problems in their static setting in truly subquadratic time thus refuting SETH. 

Other applications of the same scheme are also presented. We propose data structures using $\Ohtilde(n)$ space that compute the following characteristics of a single string $S$ after any edit operation:
\begin{itemize}
  \item longest palindrome substring of $S$ in $\Ohtilde(\sqrt{n})$ time;
  \item longest substring of $S$ being a Lyndon string as well as a representation of the Lyndon factorization of $S$ that allows, in particular, to extract the $i$th factor in $\Oh(\log^2 n)$ time in $\Ohtilde(\sqrt{n})$ time;
  \item longest repeat, that is, longest substring that occurs more than once in $S$ in $\Ohtilde(n^{2/3})$ time.
\end{itemize}

We state our results for strings of length at most $n$, over an integer alphabet $\Sigma = \{1,\ldots,n^{\Oh(1)}\}$. We assume the standard word RAM model with word size $\Omega(\log n)$.

\subsection{Our Techniques}
\paragraph{General Scheme and relation to Internal Pattern Matching} Our approach for most of the considered dynamic problems on strings is as follows. Let the input be a string $S$ of length $n$ (in the case of the LCS problem, this can be the concatenation of the input strings $S$ and $T$ separated by a delimiter). We construct a data structure that answers the following type of queries: given $k$ edit operations on $S$, compute the answer to a particular problem on the resulting string $S'$. Assuming that the data structure occupies $\cO(s_n)$ space, answers queries for $k$ edits in time $\cO(q_n(k))$ and can be constructed in time $\cO(t_n)$ ($s_n \ge n$ and $q_n(k) \geq k$ is continuously non-decreasing with respect to $k$), this data structure can be used to design a dynamic algorithm that preprocesses the input string in time $\cO(t_n)$ and answers queries dynamically under edit operations in amortized time $\cO(q_n(\kappa))$, where $\kappa$ is such that $q_n(\kappa)=(t_n+n)/\kappa$, using $\cO(s_n)$ space. The query complexity can be made worst-case using the technique of time slicing. In particular, for $s_n,t_n=\Ohtilde(n)$ and $q_n(k)=\Ohtilde(k)$ we obtain a fully dynamic algorithm with $\Ohtilde(\sqrt{n})$-time queries, whereas for $q_n(k)=\Ohtilde(k^2)$ the query time is $\Ohtilde(n^{2/3})$.

A \emph{$k$-substring} of a string $S$ is a concatenation of $k$ strings, each of which is either a substring of $S$ or a single character. A $k$-substring of $S$ can be represented on a doubly-linked list in $\Oh(k)$ additional space if the string $S$ itself is stored. The string $S$ after $k$ subsequent edit operations can be represented as a $(2k+1)$-substring due to the following observation.

\begin{lemma}\label{lem:ksub}
  Let $S'$ be a $k$-substring of $S$ and $S''$ be $S'$ after a single edit operation. Then $S''$ is a $(k+2)$-substring of $S$. Moreover, $S''$ can be computed from $S'$ in $\Oh(k)$ time.
\end{lemma}
\begin{proof}
  Let $S'=F_1 \dots F_k$ where each $F_i$ is either a substring of $S$ or a single character. We traverse the list of substrings until we find the substring $F_i$ such that the edit operation takes place at the $j$-th character of $F_i$. As a result, $F_i$ is decomposed into a prefix and a suffix, potentially with a single character inserted in between. The resulting string $S''$ is a $(k+2)$-substring of $S$.
\end{proof}

Thus the fully dynamic version reduces to designing a data structure over a string $S$ of length~$n$ that computes the result of a specific problem on a $k$-substring $F_1 \dots F_k$ of $S$. For the considered problems we aim at computing the longest substring of $S$ that satisfies a certain property. Then there are two cases. Case 1: the sought substring occurs inside one of the substrings $F_i$ (or each of its two occurrences satisfies this property in case of the LCS and the longest repeat problems). Case 2: it contains the boundary between some two substrings $F_i$ and $F_{i+1}$.

Case 1 requires to compute the solution to a certain problem on a substring or substrings of a specified string. This is the so-called \emph{internal} model of queries; this name was coined in the paper of Kociumaka et al.~\cite{DBLP:conf/soda/KociumakaRRW15}. We call case 2 \emph{cross-substring} queries. As it turns out, certain internal queries arise in cross-substring queries as well due to string periodicity.

\paragraph{Internal Queries for LCS}
Let us first consider LCS queries in this model. More precisely, we are given the strings $S$ and $T$ and we are to answer LCS queries between a substring of $S$ and a substring of $T$.

The most general internal LCS query can be reduced via a binary search to two-range-LCP queries of Amir et al.~\cite{DBLP:conf/spire/AmirLT15}. With their Theorem~6, one can construct a data structure of size $\cO(n)$ in $\cO(n\sqrt{n})$ time that allows for $\Ohtilde(\sqrt{n})$-time queries.
We cannot use this data structure in our scheme due to its high preprocessing cost. 
In fact, Amir et al.~\cite{DBLP:conf/spire/AmirLT15} show that the two-range-LCP data structure problem is at least as hard as the \emph{Set Emptiness} problem, where one is to preprocess a collection of sets of total cardinality $n$ so that queries of whether the intersection of two sets is empty can be answered efficiently.
The best known $\cO(n)$-sized data structure for this problem has $\cO(\sqrt{n/w})$-query-time, where $w$ is the size of the computer word.
It is not hard to see that the reduction of~\cite{DBLP:conf/spire/AmirLT15} can be adapted to show that answering general internal LCS queries is at least as hard as answering Set Emptiness queries.
In light of this, we develop a different global approach to circumvent answering such queries.

We rely on an algorithm that we develop for a particular decremental version of the LCS problem, where characters in $S$ (resp.~in $T$) are dynamically replaced by $\# \notin \Sigma$ (resp.~$\$ \notin \Sigma$), with $\# \neq \$$.
We decompose this problem into two cases, depending on whether the sought LCS is short or long.
We solve the first case by using identifiers for all sufficiently short substrings of $S$ and $T$ based on the heavy path that the node that represents them in the generalized suffix tree of the two strings lies on. 
We tackle the second case by a non-trivial technique that employs difference covers.
A difference $d$-cover for $\{1,\ldots,n\}$ is a subset $A \subseteq \{1,\ldots,n\}$ of size $\cO(n/\sqrt{d})$ that can be constructed efficiently (\cite{DBLP:journals/tocs/Maekawa85}), and guarantees that, if an LCS of length $\ell \geq d$ occurs at position $i$ in $S$ and position $j$ in $T$, there is a pair of elements $p,q \in A$ such that $0 \leq p-i=q-j < \ell$.
Informally, these occurrences of an LCS are anchored at the pair $(p,q)$.
We show how to exploit this to construct a compact trie for a family of strings of size proportional to the size of the difference cover.
Then we use a data structure, that was shown in~\cite{DBLP:conf/cpm/Charalampopoulos18} (and, implicitly, in \cite{DBLP:journals/tcs/CrochemoreIMS06,DBLP:journals/ipl/FlouriGKU15}), to answer a certain kind of longest common prefix queries over that trie.

We also design data structures based (mostly) on orthogonal range maximum queries with $\Ohtilde(n)$ size and construction time and $\Ohtilde(1)$ queries for several special variants of the internal LCS problem that arise in the dynamic LCS problems:
\begin{itemize}
  \item LCS between a substring of $S$ and the whole $T$ --- for the dynamic problem where edits happen only in $S$;
  \item LCS between a prefix or suffix of $S$ and a prefix or suffix of $T$ --- for the problem of \textsc{LCS after One Substitution per String};
  \item Given three substrings $U$, $V$, and $W$ of $T$, compute the longest substring $XY$ of $W$ such that $X$ is a suffix of $U$ and $Y$ is a prefix of $V$, that we call \textsc{Three Substrings LCS} queries --- that are used ($k$ times) in the computation of the LCS between a $k$-substring of $S$ and a substring of $T$.
\end{itemize}
A solution to a special case of \textsc{Three Substrings LCS} queries with $W=T$ was already implicitly present in Amir et al.~\cite{Amir2017}. It is based on the heaviest induced ancestors (HIA) problem on trees, introduced by Gagie et al.~\cite{DBLP:conf/cccg/GagieGN13}, applied to the suffix tree of $T$. We generalize the HIA queries and use them to answer general \textsc{Three Substrings LCS} queries.

The data structure for answering our generalization of HIA queries turns out to be one of the most technical parts of the paper.
It relies on the construction of multidimensional grids for pairs of heavy paths of the involved trees.
Each query can be answered by interpreting the answer of $\cO(\log^2 n)$ orthogonal range maximum queries over such grids.

\paragraph{Internal Queries for Palindromes} We show a data structure with $\Ohtilde(n)$ preprocessing time and $\Ohtilde(1)$ time for internal longest palindrome substring queries. As previously, it is based on orthogonal range maximum queries.

\paragraph{Internal Queries for Lyndon Strings} Urabe et al.~\cite{DBLP:conf/cpm/UrabeNIBT18} show how to compute a representation of a Lyndon factorization of a prefix of a string and of a suffix of a string in $\Ohtilde(1)$ time after $\Ohtilde(n)$ preprocessing. For the prefixes, their solution is based on the Lyndon representations of prefixes of a Lyndon string, whereas for the suffixes, the structure of a Lyndon tree (due to \cite{DBLP:journals/jct/Barcelo90}) is used. We combine the two approaches to obtain general internal computation of a representation of a Lyndon factorization in the same time bounds.

\paragraph{Internal Queries for String Periodicity Used} In cross-substring queries for LCS and longest palindrome we use internal Prefix-Suffix Queries of Kociumaka et al.~\cite{DBLP:conf/soda/KociumakaRRW15}. A Prefix-Suffix Query gets as input two substrings $Y$ and $Z$ in a string $S$ of length $n$ and an integer $d$ and returns the lengths of all prefixes of $Y$ of length between $d$ and $2d$ that are suffixes of $Z$. A Prefix-Suffix Query returns an arithmetic sequence. If it has at least three elements, then its difference equals the period of all the corresponding prefixes-suffixes. Kociumaka et al.~\cite{DBLP:conf/soda/KociumakaRRW15} show that Prefix-Suffix Queries can be answered in $\cO(1)$ time using a data structure of $\cO(n)$ size, which can be constructed in $\cO(n)$ time. They also show that it can be used to answer internal Period Queries, that compute the shortest period of a substring, in $\Oh(\log n)$ time.

\subsection{Roadmap}

Section~\ref{sec:prel} provides the necessary definitions and notation used throughout the paper as well the standard algorithmic toolbox for string processing.
In Section~\ref{sec:1-1} we consider a generalization of the problem from Amir et al.~\cite{Amir2017}, namely, a single substitution is allowed in each of the strings (again, both substitutions are reverted afterwards).
In Section~\ref{sec:onesided}, we consider the case where edit operations are only allowed in one of the strings.
In Section~\ref{sec:fully}, we combine these techniques and develop further ones to solve the fully dynamic LCS problem.
The missing technical details, including special cases of internal LCS queries, are left for Section~\ref{sec:internal}.
%
In Sections~\ref{sec:lr}, \ref{sec:pals} and \ref{sec:lyndon} we showcase the generality of our techniques by applying them to obtain fully dynamic algorithms for computing the longest repeat, the longest palindrome, and the longest Lyndon substring of a string (as well as a representation of its Lyndon factorization).

\section{Preliminaries}\label{sec:prel}
Let $S=S[1]S[2]\ldots S[n]$ be a \textit{string} of length $|S|=n$ over an integer  alphabet $\Sigma$. For two positions $i$ and $j$ on $S$, we denote by $S[i\dd j]=S[i]\ldots S[j]$ the substring of $S$ that starts at position $i$ and ends at position $j$ (it is empty if $j<i$). A substring of $S$ is represented in $\cO(1)$ space by specifying the indices $i$ and $j$. A prefix $S[1\dd j]$ is denoted by $S^{(j)}$ and a suffix $S[i\dd n]$ is denoted by $S_{(i)}$. A substring is called \emph{proper} if it is shorter than the whole string. We denote the {\em reverse string} of $S$ by $S^R=S[n]S[n-1]\ldots S[1]$. By $ST$, $S^k$, and $S^\infty$ we denote the concatenation of strings $S$ and $T$, $k$ copies of the string $S$, and infinitely many copies of string $S$, respectively. If a string $B$ is both a proper prefix and a proper suffix of a non-empty string $S$ of length $n$, then $B$ is called a {\em border} of~$S$. A positive integer $p$ is called a \emph{period} of $S$ if $S[i] = S[i + p]$ for all $i = 1, \ldots, n - p$. String $S$ has a period $p$ if and only if it has a border of length $n-p$. We refer to the smallest period as \emph{the period} of the string and, analogously, to the longest border as \emph{the border} of the string.

The \textit{suffix tree} $\mathcal{T}(S)$ of a non-empty string $S$ of length $n$ is a compact trie representing all suffixes of $S$. The \textit{branching} nodes of the trie as well as the \textit{terminal} nodes, that correspond to suffixes of $S$, become {\em explicit} nodes of the suffix tree, while the other nodes are {\em implicit}. Each edge of the suffix tree can be viewed as an upward maximal path of implicit nodes starting with an explicit node. Moreover, each node belongs to a unique path of that kind. Thus, each node of the trie can be represented in the suffix tree by the edge it belongs to and an index within the corresponding path. We let  $\mathcal{L}(v)$  denote the \textit{path-label} of a node $v$, i.e., the concatenation of the edge labels along the path from the root to $v$. We say that $v$ is  path-labelled  $\mathcal{L}(v)$. Additionally, $\mathcal{D}(v)= |\mathcal{L}(v)|$ is used to denote  the \textit{string-depth} of node $v$. A terminal node $v$ such that $\mathcal{L}(v) = S_{(i)}$ for some $1 \leq i \leq n$ is also labelled with index~$i$. Each substring of $S$ is uniquely represented by either an explicit or an implicit node of $\mathcal{T}(S)$, called its \emph{locus}. In standard suffix tree implementations, we assume that each node of the suffix tree is able to access its parent. Once $\mathcal{T}(S)$ is constructed, it can be traversed in a depth-first manner to compute the string-depth $\mathcal{D}(v)$ for each explicit node $v$. The suffix tree of a string of length $n$, over an integer ordered alphabet, can be computed in time and space $\cO(n)$~\cite{F97}. In the case of integer alphabets, in order to access the child of an explicit node by the first character of its edge label in $\cO(1)$ time, perfect hashing~\cite{DBLP:journals/jacm/FredmanKS84} can be used. A \emph{generalized suffix tree} (GST) of strings $S_1,\ldots,S_k$, denoted by $\mathcal{T}(S_1,\ldots,S_k)$, is the suffix tree of $X=S_1\#_1 \ldots S_k\#_k$, where $\#_1,\ldots,\#_k$ are distinct end-markers.

By $\lcpstring(S,T)$ we denote the longest common prefix of $S$ and $T$, by $\lcp(S,T)$ we denote $|\lcpstring(S,T)|$, and by \lcp{}$(r, s)$ we denote $\lcp(S_{(r)},S_{(s)})$.
A lowest common ancestor data structure can be constructed over the suffix tree in $\cO(n)$ time and $\cO(n)$ space~\cite{Bender2000} supporting \lcp{}$(r,s)$-queries in $\cO(1)$ time per query. A symmetric construction on $S^R$ can answer the so-called \emph{longest common suffix} (\lcs) queries in the same complexity. The \lcp{} and \lcs{} queries are also known as {\em longest common extension} (\LCE{}) queries. One can use the so-called \emph{kangaroo method} to compute $\LCE{}$ with a few mismatches~\cite{DBLP:journals/tcs/LandauV86}:

\begin{observation}[\cite{DBLP:journals/tcs/LandauV86}]\label{obs:kangaroo}
  Let $S$ be a string and let $S'$ be the string $S$ after $k$ substitutions. Having a data structure for answering $\LCE{}$ queries in $S$ and the left-to-right list of substitutions, one can answer $\LCE{}$ queries in $S'$ in $\Oh(k)$ time.
\end{observation}

\section{LCS After One Substitution Per String}\label{sec:1-1}
Let us now consider an extended version of the \textsc{LCS After One Substitution} problem, for simplicity restricted to substitutions.

{\defDSproblem{\textsc{LCS after One Substitution per String}}{Two strings $S$ and $T$ of length at most $n$}
{For given indices $i$, $j$ and characters $\alpha$ and $\beta$, compute $\LCS(S',T')$ where $S'$ is $S$ after substitution $S[i]:=\alpha$ and $T'$ is $T$ after substitution $T[j]:=\beta$}

In this section we prove the following result.

\begin{theorem}\label{thm:1-1}
\textsc{LCS after One Substitution Per String} can be computed in $\tilde{\cO}(1)$ time after $\tilde{\cO}(n)$-time and space preprocessing.
\end{theorem}

In the solution we consider three cases depending on whether an occurrence of the LCS contains any of the changed positions in $S$ and $T$.

\subsection{LCS Contains No Changed Position}\label{sec:1-1.1}
We use the following lemma for a special case of internal LCS queries. Its proof can be found in Section~\ref{sec:internal}.

\begin{restatable}{lemma}{prefixSuffix}\label{lem:prefix-suffix-LCS}
  Let $S$ and $T$ be two strings of length at most $n$. After $\Oh(n \log^2 n)$-time and $\Oh(n \log n)$-space preprocessing, one can compute an LCS between any prefix or suffix of $S$ and prefix or suffix of $T$ in $\Oh(\log n)$ time.
\end{restatable}

It suffices to apply internal LCS queries of Lemma~\ref{lem:prefix-suffix-LCS} four times: each time for one of $S^{(i-1)},S_{(i+1)}$ and one of $T^{(j-1)},T_{(j+1)}$. 

\subsection{LCS Contains a Changed Position in Exactly One of the Strings}\label{sec:1-1.2}
We use the following lemma that was implicitly shown in \cite{Amir2017}. In Appendix~\ref{app:ami} we give its proof for completeness. It involves computing so-called \emph{ranges} of substrings in the \emph{generalized suffix array} of $S$ and $T$.

\begin{lemma}[\cite{Amir2017}]\label{lem:merge_ranges}
  Let $S$ and $T$ be strings of length at most $n$. After $\Oh(n \log \log n)$-time and $\Oh(n)$-space preprocessing, given two substrings $U$ and $V$ of $S$ or $T$, one can (a) compute a substring of $T$ equal to $UV$, if it exists, in $\Oh(\log \log n)$ time and (b) compute the longest substring of $T$ that is equal to a prefix (or a suffix) of $UV$ in $\Oh(\log n \log \log n)$ time.
\end{lemma}

We show how to compute the longest substring that contains the position $i$ in $S$, but not the position $j$ in $T$ (the opposite case is symmetric). We first use \cref{lem:merge_ranges}(b) to compute in $\Oh(\log n \log \log n)$ time two substrings, $U$ and $V$, of $T$:
\begin{itemize}
\item $U$ is the longest substring of $T$ that is equal to a suffix of $S[1 \dd i-1]$;
\item $V$ is the longest substring of $T$ that is equal to a prefix of $\alpha S[i+1 \dd |S|]$.
\end{itemize}

Our task then reduces to computing the longest substring of $UV$ that crosses the boundary between $U$ and $V$ and is a substring of $T^{(j-1)}$ or of $T_{(j+1)}$. It can be found by two \textsc{Three substrings LCS} queries: one with $W=T^{(j-1)}$ and one with $W=T_{(j+1)}$. By the following lemma, proved in Section~\ref{sec:internal}, such queries can be answered in $\tilde{\cO}(1)$ time after $\tilde{\cO}(n)$-time preprocessing.

\begin{restatable}{lemma}{threeSubstring}\label{lem:3-substring-LCS}
  Let $T$ be a string of length at most $n$. After $\Ohtilde(n)$-time preprocessing, one can answer \textsc{Three Substrings LCS} queries in $\Ohtilde(1)$ time.\end{restatable}


\subsection{LCS Contains a Changed Position in Each of the Strings}\label{sec:1-1.3}
Recall that a Prefix-Suffix Query gets as input two substrings $Y$ and $Z$ in a string $S$ of length $n$ and an integer $d$ and returns the lengths of all prefixes of $Y$ of length between $d$ and $2d$ that are suffixes of $Z$. It is known that such a query returns an arithmetic sequence and if it has at least three elements, then its difference equals the period of all the corresponding prefixes-suffixes. Moreover, Kociumaka et al.~\cite{DBLP:conf/soda/KociumakaRRW15} show that Prefix-Suffix Queries can be answered in $\cO(1)$ time using a data structure of $\cO(n)$ size, which can be constructed in $\cO(n)$ time. By considering $X=Y=U$, this implies the two respective points of the lemma below.

\begin{lemma}\label{lem:borders}
  \begin{enumerate}[(a)]
    \item\label{la}
    For a string $U$ of length $m$, the set $\B_r(U)$ of border lengths of $U$ between $2^r$ and $2^{r+1}-1$ is an arithmetic sequence.
    If it has at least three elements, all the corresponding borders have the same period, equal to the difference of the sequence.
    \item\label{lb}
    Let $S$ be a string of length $n$. After $\Oh(n)$ time and space preprocessing, for any substring~$U$ of~$S$ and integer $r$, the arithmetic sequence $\B_r(U)$ can be computed in $\Oh(1)$ time.
  \end{enumerate}
\end{lemma}
\newcommand{\eqdef}{\ensuremath{\stackrel{\text{\tiny def}}{=}}}

Let us proceed with an algorithm that finds a longest string $S'[i_{\ell} \dd 
i_r]=T'[j_{\ell} \dd j_r]$ such that $i_{\ell} \leq i \leq i_r$ and $j_{\ell} \leq j \leq j_r$. Let us assume that $i-i_{\ell} \leq j-j_{\ell}$; the symmetric case can be treated analogously. We then have that $U\eqdef S'[i+1 \dd i_{\ell}+j-j_{\ell}-1]=T'[j_{\ell}+i-i_{\ell}+1 \dd j-1]$ as shown in Fig.~\ref{fig:both_changed_1}. Note that these substrings do not contain any of the changed positions. Further note that $U=\varepsilon$ can correspond to $i-i_{\ell} = j-j_{\ell}$ or $i-i_{\ell}+1 = j-j_{\ell}$, so both these cases need to be checked.

  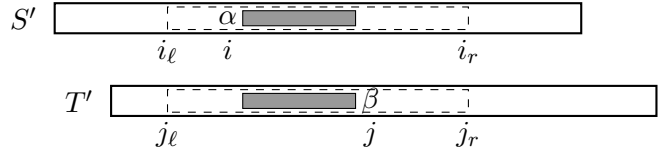
\begin{figure}[htpb]
    \begin{center}
      \begin{tikzpicture}[scale=0.5]


\begin{scope}[yshift=2.2cm]

  \draw[thick] (0,-0.4) rectangle (14,0.4);
  \draw[dashed] (3,-0.3) rectangle (11,0.3);
  \draw[fill=white!60!black,draw] (5,-0.2) rectangle (8,0.2);

  \draw (0,0) node[left=0.1cm] {$S'$};
  \draw (4.6,-0.3) node[below] {$i$};
  \draw (4.6,0) node[] {$\alpha$};
    \draw (3,-0.3) node[below] {$i_{\ell}$};
    \draw (11,-0.3) node[below] {$i_r$};

\end{scope}


   \draw[thick] (1.5,-0.4) rectangle (16,0.4);
   \draw[dashed] (3,-0.3) rectangle (11,0.3);
   \draw[fill=white!60!black,draw] (5,-0.2) rectangle (8,0.2);
   
  \draw (1.5,0) node[left=0.1cm] {$T'$};
  \draw (8.4,-0.3) node[below] {$j$};
  \draw (8.4,0) node[] {$\beta$};
    \draw (3,-0.3) node[below] {$j_{\ell}$};
    \draw (11,-0.3) node[below] {$j_r$};

\end{tikzpicture}
    \end{center}
    \caption{
      Occurrences of an LCS of $S'$ and $T'$ containing both changed positions are denoted by dashed rectangles.
      Occurrences of $U$ at which an LCS is aligned are denoted by gray rectangles.
    }\label{fig:both_changed_1}
  \end{figure}
  
Any such $U$ is a prefix of $S_{(i+1)}$ and a suffix of $T^{(j-1)}$; let $U_0$ denote the longest such string. Then, the possible candidates for $U$ are $U_0$ and all its borders. For a border $U$ of $U_0$, we say that 
\[\lcsstring(S'^{(i)},T'^{(j-|U|-1)})\, U\, \lcpstring(S'_{(i+|U|+1)},T'_{(j)})\]
is an \emph{LCS aligned at $U$}.

We compute $U_0$ in time $\cO(\log n)$ by asking Prefix-Suffix Queries for $Y=S_{(i+1)}$, $Z=T^{(j-1)}$ in $S\#T$ and $d=2^r$ for all $r=0,1, \ldots , \lfloor \log j \rfloor$. We then consider the borders of $U_0$ in arithmetic sequences of their lengths; see Lemma~\ref{lem:borders} in this regard.

If an arithmetic sequence has at most two elements, we compute an LCS aligned at each of the borders in $\Oh(1)$ time by the formula above using the kangaroo method (Observation~\ref{obs:kangaroo}). Otherwise, let $p$ be the difference of the arithmetic sequence, $\ell$ be its length, and $u$ be its maximum element. Let us further denote:

\vspace*{-0.3cm}
\begin{align*}
  X_1=S'_{(i+u+1)},\quad
  Y_1=T'_{(j)},\quad
  P_1=S'[i+u-p+1\dd i+u],\\
  X_2^R=T'^{(j-u-1)},\quad
  Y_2^R=S'^{(i)},\quad
  P_2^R=T'[j-u\dd j-u+p-1].
\end{align*}
The setting is presented in~\cref{fig:both_changed_2}.

It can be readily verified (inspect~\cref{fig:both_changed_2}) that a longest common substring aligned at the border of length $u - wp$, for $w \in [0,\ell -1]$, is equal to
\[g(w)=\lcs(X_2^R (P_2^R)^w,Y_2^R) + u - wp + \lcp(P_1^w X_1,Y_1) = \lcp(P_2^w X_2,Y_2) + \lcp(P_1^w X_1,Y_1) + u - wp.\]
Thus, a longest LCS aligned at a border whose length is in this arithmetic sequence is $\max_{w=0}^{\ell-1} g(w)$.

    \begin{figure}[htpb]
    \begin{center}
      \begin{tikzpicture}[scale=0.5]


\begin{scope}[yshift=5cm,xshift=1cm]

  \draw[thick] (0,-0.5) rectangle (15,0.5);
  \draw[fill=white!60!black,draw] (5,-0.5) rectangle (12.4,0.5);

  \draw (0,0) node[left=0.1cm] {$S'$};
  \draw (4.6,-0.4) node[below] {$i$};
  \draw (4.6,0) node[] {$\alpha$};
      \draw (2.5,0) node {$Y_2^R$};
      \draw (13.7,0) node {$X_1$};
  \draw (11.4,-0.5) -- (11.4,0.5);
  \draw (11.9,0) node {$P_1$};
   			   \draw (10.4,-0.5) -- (10.4,0.5);
  				\draw (10.9,0) node {$P_1$};
  				\draw (9.4,-0.5) -- (9.4,0.5);
  				\draw (9.9,0) node {$P_1$};
    \draw[latex-] (5,1.4) -- (8.2,1.4);  
        \draw (8.7,1.4) node {$u$};
  \draw[-latex] (9.2,1.4) -- (12.4,1.4);

    \draw[latex-] (9.4,-1) -- (10.4,-1);
    	\draw (10.9,-1) node {$3p$};
     \draw[-latex] (11.4,-1) -- (12.4,-1);

\begin{scope}[yshift=0.6cm]
  \foreach \x in {0,1,2,3,4,5,6,7,8,9,10,11}{
    \clip (5,0) rectangle (12.4,0.7);
    \draw[xshift=\x cm] (5,0) sin (5.5,0.5) cos (6,0);
  }
  
\end{scope}

\end{scope}


   \draw[thick] (0,-0.5) rectangle (16,0.5);
   \draw[fill=white!60!black,draw] (3,-0.5) rectangle (10.4,0.5);
   
  \draw (0,0) node[left=0.1cm] {$T'$};
      \draw (1.5,0) node {$X_2^R$};
      \draw (4,-0.5) -- (4,0.5);
      \draw (3.5,0) node {$P_2^R$};
      		 \draw (5,-0.5) -- (5,0.5);
     		 \draw (4.5,0) node {$P_2^R$};     
     		   \draw (6,-0.5) -- (6,0.5);
    		  \draw (5.5,0) node {$P_2^R$};    
      \draw (13.15,0) node {$Y_1$};
  \draw (10.8,-0.4) node[below] {$j$};
  \draw (10.8,0) node[] {$\beta$};
  
    \draw[latex-] (3,1.6) -- (4,1.6);
    	\draw (4.5,1.6) node {$3p$};
     \draw[-latex] (5,1.6) -- (6,1.6);  
  
\begin{scope}[yshift=0.6cm]
  \foreach \x in {0,1,2,3,4,5,6,7,8,9,10,11}{
    \clip (3,0) rectangle (10.4,0.7);
    \draw[xshift=\x cm] (3,0) sin (3.5,0.5) cos (4,0);
  }
\end{scope}


\begin{scope}[yshift=2.8cm]
  \draw[dashed] (1.5,-0.5) rectangle (13,0.5);
  \draw[dashed,fill=white!80!black,draw] (6,-0.5) rectangle (10.4,0.5);
    \draw[latex-] (6,0) -- (7,0);
    	\draw (8.2,0) node {$u-3p$};
     \draw[-latex] (9.4,0) -- (10.4,0);  
\end{scope}
\end{tikzpicture}
    \end{center}
    \caption{
      The border of length $u$ is denoted by gray rectangles. An LCS aligned at the border of length $u-3p$, which is in the same arithmetic sequence, is denoted by the dashed rectangle.
    }\label{fig:both_changed_2}
  \end{figure}
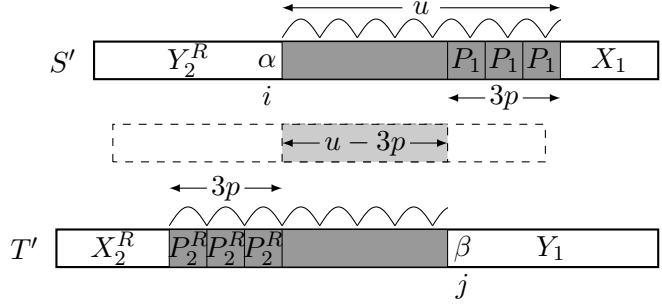

The following observation facilitates efficient evaluation of this formula.
\begin{observation}\label{obs:Pinfty}
  For any strings $P,X,Y$, the function $f(w)=\lcp(P^w X, Y)$ for integer $w \ge 0$ is piecewise linear with at most three pieces. Moreover, if $P,X,Y$ are substrings of a string $S$, then the exact formula of $f$ can be computed with $\Oh(1)$ \LCE{} queries on $S$.
\end{observation}
\begin{proof}
  Let $a=\lcp(P^\infty,X)$, $b=\lcp(P^\infty,Y)$, and $p=|P|$. Then:
  \[f(w)=\left\{\begin{array}{ll}
    a+wp & \text{if }a+wp<b\\
    w+\lcp(X,Y[aw+1 \dd |Y|]) & \text{if }a+wp=b\\
    b & \text{if }a+wp>b.
  \end{array} \right .\]

  Note that $a$ can be computed from $\lcp(P,X)$ and $\lcp(X,X[p+1 \dd |X|])$, and $b$ analogously. Thus if $P,X,Y$ are substrings of $S$, five \LCE{} queries on $S$ suffice.
\end{proof}

By Observation~\ref{obs:Pinfty}, $g(w)$ can be expressed as a piecewise linear function with $\Oh(1)$ pieces. Moreover, its exact formula can be computed using $\Oh(1)$ \LCE{} queries on $S'\#T'$, hence, in $\Oh(1)$ time using the kangaroo method (Observation~\ref{obs:kangaroo}). This allows to compute $\max_{w=0}^{\ell-1} g(w)$ in $\Oh(1)$ time.

Each arithmetic sequence is processed in $\cO(1)$ time. The global maximum that contains both changed positions is the required answer. Thus the query time in this case is $\Oh(\log n)$ and the preprocessing requires $\Oh(n)$ time and space.

By combining the results of~\cref{sec:1-1.1,sec:1-1.2,sec:1-1.3}, we arrive at Theorem~\ref{thm:1-1}.







\section{One-Sided Fully Dynamic LCS}\label{sec:onesided}

In this section we show a solution to fully dynamic LCS in the case that edit operations are only allowed in $S$. We consider an auxiliary problem.

{\defDSproblem{$k$\textsc{-Substring LCS}}{Two strings $S$ and $T$ of length at most $n$}{Compute $\LCS(S',T)$ where $S'=F_1 \dots F_k$ is a $k$-substring of $S$}

\subsection{Internal Queries}
It suffices to apply internal LCS queries for a substring of $S$ and the whole $T$. The following lemma is proved in Section~\ref{sec:internal}.

\begin{restatable}{lemma}{substringString}\label{lem:substring-string-LCS}
  Let $S$ and $T$ be two strings of length at most $n$. After $\cO(n)$-time preprocessing, one can compute an LCS between $T$ and any substring of $S$ in $\Oh(\log n)$ time.
\end{restatable}

We need to answer $k$ such queries. Additionally, the case that some $F_i$ is a single character can be served by storing a balanced BST over all characters in $T$.

\subsection{Cross-Substring Queries}
We now present an algorithm that computes, for each $i=1,\ldots,k$, the longest substring of $S'$ that contains the first character of $F_i$, but not of $F_{i-1}$ and occurs in $T$ in $\cO(k \log^2 n)$ time. These are the possible LCSs that cross the substring boundaries.

For convenience let us assume that $F_0=F_{k+1}$ are empty strings. For an index $i \in \{1,\ldots,k\}$, by $\nnext(i)$ we denote the greatest index $j \ge i-1$ for which $F_i \dots F_j$ is a substring of $T$. These values are computing using a sliding-window-based approach. We start with computing $\nnext(1)$. To this end, we use Lemma~\ref{lem:merge_ranges}(a) for subsequent substrings $F_1,F_2, \ldots$ as long as their concatenation is a substring of $T$. This takes $\Oh(k \log \log n)$ time. Now assume that we have computed $\nnext(i)$ and we wish to compute $\nnext(i+1)$. Obviously, $\nnext(i+1) \ge \nnext(i)$. Let $j=\nnext(i)$. We start with $F_{i+1} \dots F_j$ which is represented as a substring of $T$. We keep extending this substring by $F_{j+1},F_{j+2},\ldots$ using Lemma~\ref{lem:merge_ranges}(a) as before as long as the concatenation is a substring of $T$. In total, computing values $\nnext(i)$ for all $i=1,\ldots,k$ takes $\Oh(k \log \log n)$ time.

Let us now fix $i$ and let $j=\nnext(i)$. We use Lemma~\ref{lem:merge_ranges}(b) to find the longest prefix $P_i$ of $(F_i \dots F_j)F_{j+1}$ that occurs in $T$; it is also the longest prefix of $F_i \dots F_k$ that occurs in $T$ by the definition of $\nnext(i)$. Then Lemma~\ref{lem:merge_ranges}(b) can be used to compute the longest suffix $Q_i$ of $F_{i-1}$ that occurs in $T$. We then compute the result by a \textsc{Three Substrings LCS} query for $U=Q_i$, $V=P_i$, and $W=T$ using a special case of \textsc{Three Substrings LCS} queries.

\begin{restatable}{lemma}{threeSubstringsSpecial}\label{lem:3-substring-LCS-special}
  The \textsc{Three Substrings LCS} queries with $W=T$ can be answered in $\cO(\log^2 n)$ time, using a data structure of size $\cO(n\log^2 n)$ that can be constructed in $\cO(n\log^2 n)$ time.
\end{restatable}

Lemma~\ref{lem:3-substring-LCS-special} was implicitly proved by Amir et al.\ in~\cite{Amir2017}. In Section~\ref{sec:internal} we sketch its proof for completeness.

\textit{Complexity.} Recall that computing the values $\nnext(i)$ takes time $\cO(k \log \log n)$. For each $i$ it takes time $\cO(\log n \log \log n)$ to find $P_i$ and $Q_i$. These steps require $\Oh(n\log \log n)$-time and $\Oh(n)$-space preprocessing. Finally, a \textsc{Three substrings LCS} query with $W=T$ can be answered in $\Oh(\log^2 n)$ time with a data structure of size $\cO(n\log^2 n)$ that can be constructed in $\cO(n\log^2 n)$ time (Lemma~\ref{lem:3-substring-LCS-special}).

\subsection{Round-Up}
Combining the results of the two subsections we obtain the following.

\begin{lemma}\label{thm:ksubs}
$k\textsc{-Substring LCS}$ queries can be answered in $\Oh(k \log^2 n)$ time, using a data structure of size $\Oh(n\log^2 n)$ that can be constructed in $\Oh(n\log^2 n)$ time.
\end{lemma}

We now formalize the time slicing deamortization technique for our purposes.

\begin{lemma}\label{lem:timeslice}
Assume that there is a data structure $\mathcal{D}$ over an input string of length $n$ that occupies $\cO(s_n)$ space, answers queries for $k$-substrings in time $\cO(q_n(k))$ and can be constructed in time $\cO(t_n)$.
Assume that $s_n \ge n$ and $q(k,n) \geq k$ is continuously non-decreasing with respect to $k$.
We can then design an algorithm that preprocesses the input string in time $\cO(t_n)$ and answers queries dynamically under edit operations in worst-case time $\cO(q_n(\kappa))$, where $\kappa$ is such that $q_n(\kappa)=(t_n+n)/\kappa$, using $\cO(s_n)$ space.
\end{lemma}
\begin{proof}
We first build $\mathcal{D}$ for the input string. The $k$-substring of $S$ after the subsequent edit operations is stored. We keep a counter $C$ of the number of queries answered since the point in time to which our data structure refers; if $C \leq 2\kappa$ and a new edit operation occurs, we create the $(2C+1)$-substring representing the string from the $(2C-1)$-substring we have using Lemma~\ref{lem:ksub} in time $\cO(C)=\cO(\kappa)$ and answer the query in time $\cO(q_n(C))=\cO(q_n(\kappa))$.

For convenience let us assume that $\kappa$ is an integer.
As soon as $C=\kappa$, we start recomputing the data structure $\mathcal{D}$ for the string after all the edit operations so far, but we allocate this computation so that it happens while answering the next $\kappa$ queries. First we create a working copy of the string in $\Oh(n)$ time and then construct the data structure in $\Oh(t_n)$ time.

When $C=2\kappa$, we set $C$ to $\kappa$, dispose of the original data structure and string and start using the new ones for the next $\kappa$ queries while computing the next one.

The following invariant is kept throughout the execution of the algorithm: the data structure being used refers to the string(s) at most $2\kappa$ edit operations before. Hence, the query time is $\cO(q_n(\kappa))$. The extra time spent during each query for computing the data structure is also $\cO((t_n+n)\kappa)=\cO(q_n(\kappa))$ since the $\cO(t_n+n)$ time is split equally among $\kappa$ queries. At every point in the algorithm we store at most two copies of the data structure. The statement follows. 
\end{proof} 

We plug Lemma~\ref{thm:ksubs} into our general scheme for dynamic string algorithms using~\cref{lem:timeslice} to obtain the first result.

\begin{theorem}\label{thm:1sided}
Fully dynamic LCS queries on strings of length up to $n$ with edits in one of the strings only can be answered in $\cO(\sqrt{n} \log^2 n)$ time per query, using $\cO(n \log n)$ space, after $\cO(n \log^2 n)$-time preprocessing.
\end{theorem}

\section{Fully Dynamic LCS}\label{sec:fully}

In this section we assume that the sought LCS has length at least $2$. The case that it is of unit or zero length can be easily treated separately.

We use the following auxiliary problem that generalizes of \textsc{LCS after One Substitution per String} into the case of $k$ substitutions:

{\defDSproblem{$(k_1,k_2)$\textsc{-Substring LCS}}{Two strings $S$ and $T$ of length at most $n$}
{Compute $\LCS(S',T')$ where $S'=F_1\ldots F_{k_1}$ is a $k_1$-substring of $S$, $T'=G_1\ldots G_{k_2}$ is a $k_2$-substring of $T$, and $k_1+k_2=k$}

As in the previous section, we consider three cases.
\begin{itemize}
\item An LCS does not contain any position (or boundary between positions) in $S$ or $T$ where an edit took place.\footnote{Considering boundaries between positions allows us to account for deletions.} As it was mentioned before, this problem probably cannot be solved efficiently in the language of $k$-substrings. Instead, we compute such an LCS via an inherently dynamic algorithm for the \textsc{Decremental LCS} problem. See~\cref{sec:fully.1}.
\item An LCS contains at least one position where an edit operation took place in exactly one of the strings. This corresponds to the $(k_1,k_2)$\textsc{-Substring LCS} problem when an LCS contains the boundary between some substrings of exactly one of $S'$ and $T'$. We compute such an LCS by combining the techniques of~\cref{sec:onesided,sec:1-1.2}. See~\cref{sec:fully.2}.
\item An LCS contains at least one position where an edit operation took place in both of the strings. This corresponds to the $(k_1,k_2)$\textsc{-Substring LCS} problem when an LCS contains the boundary between some substrings in both of $S'$ and $T'$. We compute such an LCS by building upon the techniques of~\cref{sec:1-1.3} and showing how to answer \LCE{} queries for $k$-substrings in time $\Ohtilde(1)$ after appropriate preprocessing. See~\cref{sec:fully.3}.
\end{itemize}

\subsection{Decremental LCS}\label{sec:fully.1}

In this subsection we only discuss substitutions for ease of presentation.
We use the following convenient formulation of the problem, where characters in $S$ (resp.~in $T$) are dynamically replaced by $\# \notin \Sigma$ (resp.~$\$ \notin \Sigma$), with $\# \neq \$$.

{\defDSproblem{\textsc{Decremental LCS}}{Two strings $S$ and $T$ of length at most $n$}{For a given $i$, set $S[i]:=\#$ or $T[i]:=\$$ ($\#,\$ \notin \Sigma$, $\# \neq \$$) and return $\LCS(S,T)$}

\begin{remark}
If we were to consider insertions and deletions as well, we could instead put artificial separators between positions of $S$ and $T$ and require the LCS not to contain any of the separators. I.e., an insertion at position $i$ would be handled by a separator between $S[i-1]$ and $S[i]$, and a deletion or substitution at position $i$ would be handled by separators between positions $S[i-1]$ and $S[i]$ and positions $S[i]$ and $S[i+1]$.
\end{remark}

We also use an auxiliary problem.

{\defDSproblem{\textsc{Decremental Bounded-Length LCS}}{Two strings $S$ and $T$ of length at most $n$ and positive integer $d$}{For a given $i$, set $S[i]:=\#$ or $T[i]:=\$$ ($\#,\$ \notin \Sigma$, $\# \neq \$$) and return the longest common substring of $S$ and $T$ of length at most $d$}

If $S[p]$ is replaced by a $\#$, then every substring $S[a \dd b]$ for $a \le p \le b$ can no longer be an occurrence of the LCS of $S$ and $T$. We call every such substring occurrence \emph{invalidated}. The same name is used for substring occurrences that are invalidated by character replacements in $T$. We say that a substring of $S$ (resp.~$T$) is $S$-invalidated (resp.~$T$-invalidated) if all of its occurrences in $S$ (resp.~$T$) are invalidated. We could solve the \textsc{Decremental Bounded-Length LCS} problem in $\Ohtilde(d^2)$ time per operation if we assigned a distinct integer identifier to each distinct substring of $S$ and $T$ of length $\ell \leq d$ and maintained counters on how many times each of these strings occurs in $S$ and $T$ and on how many substrings of length $\ell$ are common to both strings. Indeed, each substitution can invalidate $\Theta(d^2)$ fragments of the string it is applied to. However, this time complexity can be improved to $\Ohtilde(d)$. To this end, the identifier we give to each substring is the heavy path its locus lies on in the generalized suffix tree of $S$ and $T$. Then after each operation we have to update information corresponding to $\cO(d \log n)$ heavy paths. We show how to update the counters efficiently using interval trees. The following lemma is proved in Section~\ref{sec:internal}.
\begin{restatable}{lemma}{boundedLength}\label{lem:aux_approx}
\textsc{Decremental Bounded-Length LCS} can be solved in $\Ohtilde(d)$ time per operation after $\Ohtilde(n)$-time preprocessing, using $\Ohtilde(n+kd)$ space for $k$ performed operations. 
\end{restatable}

Let $S'$ and $T'$ be the strings $S$ and $T$ after $p$ substitutions; for some $p \le k$. For a position $i$, by $\nn_{S'}(i)$ we denote the smallest position $j \ge i$ such that $S'[j]=\#$. If no such position exists, we set $\nn_{S'}(i)=|S'|+1$. Similarly, by $\pp_{S'}(i)$ we denote the greatest position $j \le i$ such that $S'[j]=\#$, or $0$ if no such position exists. Similarly we define $\nn_{T'}(i)$ and $\pp_{T'}(i)$. Such values can be easily computed in $\Oh(\log n)$ time if the set of replaced positions is stored in a balanced BST. 

We say that a set $\S(d)\subseteq \mathbb{Z}_+$  is a $d$-\emph{cover} if
there is a constant-time computable function $h$ such that for $i,j\in \mathbb{Z}_+$ we have $0\le h(i,j)< d$ and $i+h(i,j), j+h(i,j)\in \S(d)$.

\begin{lemma}[\cite{DBLP:journals/tocs/Maekawa85,BurkhardtEtAl2003}]\label{fct:dcov}
For each $d\in \mathbb{Z}_+$ there is a $d$-cover $\S(d)$
such that $\S(d)\cap [1, n]$ is of size $\cO(\frac{n}{\sqrt{d}})$ and can be constructed in $\cO(\frac{n}{\sqrt{d}})$ time.
\end{lemma}

\begin{figure}[ht]
\begin{center}
\begin{tikzpicture}[scale=0.5]
\tikzstyle{s} = [draw, circle, fill=black, minimum size = 3pt, inner sep = 0 pt, color=black]
\draw (1,0)--(20,0);
 
\foreach \i in {1,...,20} {
  \draw (\i,-0.1)--+(0,0.2);
  \node at (\i, 0.2) [above] {\small \i};
}
\foreach \i in {2,3,5,8,9,11,14,15,17,20} {
  \node [s] at (\i,0) {};
}

\foreach \i/\l in {4/4, 11/4} {
  \draw [dotted] (\i,0)--+(0,-1.1);
  \draw [dotted] (\i+\l,0)--+(0,-1.1);
  \draw [-latex] (\i,-.8)--node[midway,below] {$h(4,11)=4$} +(\l,0);
}
\end{tikzpicture}
\end{center}

\caption{An example of a 6-cover $\S_{20}(6)=\{2,3,5, 8,9,11, 14,15,17, 20\}$, with the
elements marked as black circles. For example, we may have $h(4,11)=4$ since $4+4,\,11+4\in \S_{20}(6)$. }
\label{fig:diff_cover_example}
\end{figure}
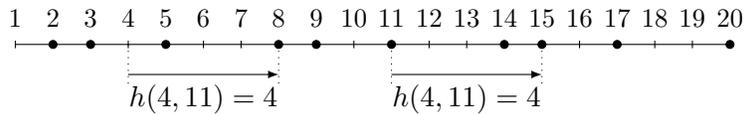

The intuition behind applying the $d$-cover in our string-processing setting is as follows (inspect also Fig.~\ref{fig:diff_cover_example}). Consider a position $i$ on $S$ and a position $j$ on $T$. Note that $i,j\in[1,n]$. By the $d$-cover construction, we have that $h(i, j)$ is within distance $d$ and $i + h(i, j), j + h(i, j) \in \S(d)$. Thus if we want to find a longest common substring of length {\em at least $d$}, it suffices to compute longest common extensions to the left and to the right {\em of only} positions $i', j' \in \S(d)$ (black circles in Fig.~\ref{fig:diff_cover_example}) and then merge these partial results accordingly.

For this we use the following auxiliary problem that was introduced in \cite{DBLP:conf/cpm/Charalampopoulos18}. 

\defproblem{\problem}{
A compact trie $\T(\F)$ of a family of strings $\F$
and two sets $\P,\Q\subseteq \F^2$}
 {The value $\maxPairLCP(\P,\Q)$, defined as\\
{$\maxPairLCP(\P,\Q)\!=\!\max\{\lcp(P_1,Q_1)+\lcp(P_2,Q_2) : (P_1,P_2)\in \P \text{ and }(Q_1,Q_2)\in \Q\}$}%
 }
 
An efficient solution to this problem was shown in~\cite{DBLP:conf/cpm/Charalampopoulos18} (and, implicitly, in \cite{DBLP:journals/tcs/CrochemoreIMS06,DBLP:journals/ipl/FlouriGKU15}).
\begin{lemma}[\cite{DBLP:conf/cpm/Charalampopoulos18}]\label{lem:problem}
\problem can be solved in $\cO(|\F|+N\log N)$ time, where $N=|\P|+|\Q|$.
\end{lemma}
\begin{lemma}\label{lem:exactdec}
\textsc{Decremental LCS} can be solved in $\Ohtilde(n^{2/3})$ time per query, using $\tilde{\cO}(n+kn^{2/3})$ space, after $\tilde{\cO}(n)$-time preprocessing for $k$ performed operations.
\end{lemma}
\begin{proof}
  Let us consider an integer $d \in [1,n]$. For lengths up to $d$, we use the algorithm for the \textsc{Decremental Bounded-Length LCS} problem of Lemma~\ref{lem:aux_approx}. If the \textsc{Decremental Bounded-Length LCS} problem indicates that there is a solution of length at least $d$, we proceed to the second step. Let $A=\S(d)\cap [1,n]$ be a $d$-cover of size $\cO(n/\sqrt{d})$ (see Lemma~\ref{fct:dcov}).
  
  We consider the following families of pairs of strings:
  
  \vspace*{-0.3cm}
  \begin{align*}
    \P &= \{\,(S[\pp_{S'}(i-1)+1 \dd i-1])^R,\,S[i \dd \nn_{S'}(i)-1])\,:\,i \in A\,\}\\
    \Q &= \{\,(T[\pp_{T'}(i-1)+1 \dd i-1])^R,\,T[i \dd \nn_{T'}(i)-1])\,:\,i \in A\,\}.
  \end{align*}
  We define $\F$ as the family of strings that occur in the pairs from $\P$ and $\Q$. Then $\maxPairLCP(\P,\Q)$ equals the length of the sought LCS, provided that it is at least $d$.
  
  Note that $|\P|,|\Q|,|\F|$ are $\cO(n/\sqrt{d})$. A compact trie $\T(\F)$ can be constructed in $\cO(|\F| \log |\F|)$ time by sorting all the strings (using $\lcp$-queries) and then a standard left-to-right construction; see~\cite{AlgorithmsOnStrings}. Thus we can use the solution to \problem which takes $\Ohtilde(n/\sqrt{d})$ time. We set $d=\lfloor n^{2/3} \rfloor$ to obtain the stated complexity.
\end{proof}

\subsection{One-Sided Cross-Substring Queries}\label{sec:fully.2}

We show a solution with $\Ohtilde(k^2)$-time queries after $\Ohtilde(n)$-time preprocessing by combining the techniques from Sections~\ref{sec:onesided} and~\ref{sec:1-1.2}. Assume that an LCS is to contain the boundary between two substrings $F_i$ and $F_{i+1}$ of $S'$; the opposite case is symmetric.
For all $i \in [1,k_1]$, we compute $\nnext(i)$, $P_i$ and $Q_i$ for $S'$ and $T$ in $\Oh(k_1 \log n \log \log n)$ time as in Section~\ref{sec:onesided}.
For each $i$, we wish to find the longest substring of $Q_iP_i$ that occurs in $T$, not containing any of the boundaries between $G_j,G_{j+1}$ for any $j$.
To this end, we perform  $k_2$ \textsc{Three Substrings LCS} queries, setting $U=Q_i$, $V=P_i$, and $W=G_j$ for each $j=1,\ldots , k_2$, if $|G_j| \ge 2$.
Recall that, by~\cref{lem:3-substring-LCS}, such queries can be answered in $\Ohtilde(1)$ time after $\Ohtilde(n)$-time preprocessing.

\subsection{Two-Sided Cross-Substring Queries}\label{sec:fully.3}

We show a solution with $\Oh(k^2 \log^3 n)$-time queries after $\Oh(n \log n)$-time preprocessing by combining the ideas presented in~\cref{sec:1-1.3} and an algorithm for efficient \LCE{} queries in a dynamic setting.
We consider each pair of boundaries between pairs $(F_i,F_{i+1})$ and $(G_j,G_{j+1})$, for $1 \leq i\leq k_1-1$ and $1 \leq j\leq k_2-1$.
We process the prefixes of $F_{i+1}$ that are suffixes of $G_j$ as in~\cref{sec:1-1.3} (the symmetric case is treated analogously).
The only difference is that we cannot answer $\LCE{}$ queries in $\cO(1)$ time in this setting. To this end we develop an $\Ohtilde(1)$ solution that relies on randomization below.

We next argue that we do not miss any possible LCS by only considering such prefix-suffix pairs of $F_{i+1}$ and $G_j$.
Let $f_i$ and $g_i$ be the starting positions of $F_{i}$ and $G_{j}$ in $S'$ and $T'$, respectively.
An LCS $S'[p \dd p+\ell-1]=T'[q \dd q+\ell-1]$ of this type will be reported when processing the pairs $(F_i,F_{i+1})$ and $(G_j,G_{j+1})$,
satisfying $p \leq f_{i+1} \leq p+\ell-1$, $q \leq g_{j+1} \leq q+\ell-1$, for which  $|(f_{i+1}-p)-(g_{j+1}-q)|$ is minimal.
Without loss of generality assume $f_{i+1}-p\leq g_{j+1}-q$.
Then, $S'[f_{i+1} \dd p+g_{j+1}-q-1]$ is a prefix of $F_{i+1}$ and $T'[q+f_{j+1}-p+1\dd g_{j+1}-1]$ is a suffix of $G_j$ and hence it is a prefix-suffix that will be processed by our algorithm; see \cref{fig:both_changed_3}.

   \begin{figure}[htpb]
    \begin{center}
      \begin{tikzpicture}[scale=0.5]


\begin{scope}[yshift=2.2cm]

  \draw[thick] (0,-0.4) rectangle (14,0.4);
  \draw[dashed] (3,-0.3) rectangle (13,0.3);
  \draw[fill=white!60!black,draw] (5,-0.2) rectangle (8,0.2);

  \draw (0,0) node[left=0.1cm] {$S'$};
  \draw (2,-0.3) node[below] {$f_{i}$};
  \draw (5.5,-0.3) node[below] {$f_{i+1}$};
    \draw (3,-0.3) node[below] {$p$};
    \draw (13,-0.3) node[below] {$p+\ell-1$};

\end{scope}


   \draw[thick] (1.5,-0.4) rectangle (16,0.4);
   \draw[dashed] (3,-0.3) rectangle (13,0.3);
   \draw[fill=white!60!black,draw] (5,-0.2) rectangle (8,0.2);
   
  \draw (1.5,0) node[left=0.1cm] {$T'$};
  \draw (2,-0.3) node[below] {$g_j$};
  \draw (8.7,-0.3) node[below] {$g_{j+1}$};
    \draw (10.4,-0.3) node[below] {$g_{j+2}$};
    \draw (3,-0.3) node[below] {$q$};
    \draw (13,-0.3) node[below] {$q+\ell-1$};

\end{tikzpicture}
    \end{center}
    \caption{
Occurrences of an LCS of $S'$ and $T'$ crossing the boundaries in both are denoted by dashed rectangles. The starting positions $f_{i+1}$ and $g_{j+1}$ minimize the formula $|(f_{i+1}-p)-(g_{j+1}-q)|$. Hence the gray rectangle, denoting $U$, is a prefix of $S'[f_{i+1} \dd f_{i+2}-1]$ and a suffix of $T'[g_{j} \dd g_{j+1}-1]$. We will thus process it as a border while processing the pairs $(F_i,F_{i+1})$ and $(G_j,G_{j+1})$ and hence find this LCS.}\label{fig:both_changed_3}
  \end{figure}
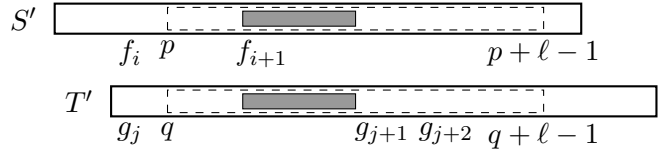
  
\LCE{} queries can be answered in time $\cO(k)$ by applying the kangaroo method.
We resort to the main result of Gawrychowski et al.~\cite{DBLP:conf/soda/GawrychowskiKKL18} (the algorithm is Las Vegas) in order to answer the queries faster. 

\begin{lemma}[Gawrychowski et al.~\cite{DBLP:conf/soda/GawrychowskiKKL18}]\label{lem:dynst}
A persistent collection $\mathcal{W}$ of strings of total length $n$ can be dynamically maintained under operations \textsf{makestring}$(W)$, \textsf{concat}$(W_1,W_2)$, \textsf{split}$(W,i)$, and \lcp$(W_1,W_2)$ with the operations requiring time $\cO(\log n +|w|)$, $\cO(\log n)$, $\cO(\log n)$ and $\cO(1)$, respectively.
\end{lemma}

We show the following lemma.

\begin{lemma}\label{lem:lce_rand}
A string $S$ of length $n$ can be preprocessed in $\cO(n)$ time and space so that $m=\cO(n)$ \LCE{} queries over a $k$-substring of $S$, $k=\cO(n)$, can be performed in $\cO(k\log n +m \log n)$ time, using this much extra space.
\end{lemma}
\begin{proof}
Our preprocessing stage consists of a single \textsf{makestring} operation, requiring time $\cO(n+\log n)=\cO(n)$.
We can construct the $k$-substring in $\cO(k \log n)$ time, using $\cO(k)$ \textsf{split}, \textsf{makestring}$(\alpha)$, $\alpha \in \Sigma$ and \textsf{concat} operations.
An \LCE{} query can be answered by two \textsf{split} operations of this $k$-substring and a single \lcp{} operation.
The length of each of the strings on which \textsf{split} and \textsf{concat} operations are performed is $\cO(n)$ and hence the cumulative length of the elements of the collection increases by $\cO(n)$ with each performed operation.
Since we perform $\cO(n)$ operations in total, the cumulative length is $\cO(n^2)$ and the time complexities of the statement follow by \cref{lem:dynst}.
\end{proof}

\begin{remark}
A lemma similar to the above, based on a simple application of Karp-Rabin fingerprints~\cite{DBLP:journals/ibmrd/KarpR87}, and avoiding the use of the heavy machinery underlying~\cref{lem:dynst}, can be proved. We leave this for the full version of this paper.
\end{remark}

We assume that $k=\cO(\sqrt{n})$, which will prove sufficient for our main result. 
We apply~\cref{lem:lce_rand} for the relevant $\Oh(k)$-substring of $S\#T$. We consider $k_1 \times k_2 = \cO(k^2)$ pairs $(F_i,F_{i+1}),(G_j,G_{j+1})$ and, by the analysis of~\cref{sec:1-1.3}, the time required for processing each of them (that is, finding the longest prefix-suffix and then considering all its borders in $\Oh(\log n)$ batches) is bounded by the time required to answer $\cO(\log n)$ $\LCE$ queries, which can be answered in time $\cO(k^2\log n)$ by the above lemma. Hence the total time required is $\cO(k^2 \log^3 n)$ after $\cO(n)$-time preprocessing.

\subsection{Main Result}\label{sec:main}

By combining the results of~\cref{sec:fully.1,sec:fully.2,sec:fully.3} and~\cref{lem:timeslice} we obtain our main result.

\begin{theorem}\label{thm:main}
Fully dynamic LCS queries on strings of length up to $n$ can be answered in $\tilde{\cO}(n^{2/3})$ time per query, using $\tilde{\cO}(n)$ space, after $\tilde{\cO}(n)$-time preprocessing.
\end{theorem}

\section{Technical Data Structures of Fully Dynamic LCS}\label{sec:internal}
Below we fill in the missing details from the dynamic LCS algorithm: simple cases of internal LCS queries in Section~\ref{sec:int1}, \textsc{Three Substrings LCS} in Section~\ref{sec:int2}, and \textsc{Decremental Bounded-Length LCS} in Section~\ref{sec:decr}. All of those data structures heavily rely on auxiliary data structures related to the suffix tree, that we recall in Section~\ref{sec:ads}. Moreover, internal LCS queries require orthogonal range maximum queries that we recall below.

Let $\mathcal{P}$ be a collection of $n$ points in a $D$-dimensional grid with integer weights and coordinates of magnitude $\cO(n)$. In a $D$-dimensional orthogonal range maximum query $\RMQ_\mathcal{P}([a_1,b_1] \times \cdots \times [a_D,b_D])$, given a hyper-rectangle $[a_1,b_1] \times \cdots \times [a_D,b_D]$, we are to report the maximum weight of a point from $\mathcal{P}$ in the rectangle. We assume that the point that attains this maximum is also computed.
The following result is known.

\begin{lemma}[\cite{DBLP:reference/cg/Agarwal04,DBLP:conf/focs/AlstrupBR00,DBLP:journals/cacm/Bentley80}]\label{lem:Ddrmq}
Orthogonal range maximum queries over a set of $n$ weighted points in $D$ dimensions, where $D = \Oh(1)$, can be answered in $\tilde{\cO}(1)$ time with a data structure of size $\tilde{\cO}(n)$ that can be constructed in $\tilde{\cO}(n)$ time. In particular, for $D=2$ one can achieve $\Oh(\log n)$ query time, $\Oh(n \log n)$ space, and $\Oh(n \log^2 n)$ construction time.
\end{lemma}

\subsection{Auxiliary Data Structures Over the Suffix Tree}\label{sec:ads}
We say that $\mathcal{T}$ is a \emph{weighted tree} if it is a rooted tree with integer weights on nodes, denoted by $\mathcal{D}(v)$, such that the weight of the root is zero and $\mathcal{D}(u) < \mathcal{D}(v)$ if $u$ is the parent of $v$. We say that a node $v$ is a \emph{weighted ancestor} of a node $u$ at depth $\ell$ if $v$ is the highest ancestor of $u$ with weight of at least $\ell$.

\begin{lemma}[\cite{DBLP:journals/talg/AmirLLS07}]\label{lem:LAQ}
After $\cO(n)$-time preprocessing, weighted ancestor queries for nodes of a weighted tree $\mathcal{T}$ of size $n$ can be answered in $\cO(\log \log n)$ time per query.
\end{lemma}

The following corollary applies Lemma~\ref{lem:LAQ} to the suffix tree.

\begin{corollary}\label{cor:LAQ}
The locus of any substring $S[i \dd j]$ in $\mathcal{T}(S)$ and can be computed in $\cO(\log \log n)$ time after $\cO(n)$-time preprocessing.
\end{corollary}

If $\mathcal{T}$ is a rooted tree, for each non-leaf node $u$ of $\mathcal{T}$ the {\em heavy edge} $(u,v)$ is an edge for which the subtree rooted at $v$ has the maximal number of leaves (in case of several such subtrees, we fix one of them). A {\em heavy path} is a maximal path of heavy edges. The path from a node $v$ in $\mathcal{T}$ to the root is composed of \emph{prefix fragments} of heavy paths interleaved by single non-heavy (compact) edges. Here a prefix fragment of a path $\pi$ is a path connecting the topmost node of $\pi$ with any of its nodes. We denote this decomposition by $H(v,\mathcal{T})$. The following observation is known; see \cite{Sleator:1983:DSD:61337.61338}.

\begin{lemma}[\cite{Sleator:1983:DSD:61337.61338}]\label{obs:hp}
  For a rooted tree $\mathcal{T}$ of size $n$ and a node $v$, $H(v,\mathcal{T})$ has size $\Oh(\log n)$ and can be computed in $\Oh(\log n)$ time after $\Oh(n)$-time preprocessing.
\end{lemma}

In the heaviest induced ancestors (HIA) problem, introduced by Gagie et al.~\cite{DBLP:conf/cccg/GagieGN13}, we are given two weighted trees $\mathcal{T}_1$ and $\mathcal{T}_2$ with the same set of $n$ leaves, numbered $1$ through $n$, and are asked queries of the following form: given a node $v_1$ of $\mathcal{T}_1$ and a node $v_2$ of $\mathcal{T}_2$, return an ancestor $u_1$ of $v_1$ and an ancestor $u_2$ of $v_2$ that have a common leaf descendant $i$ (we say that the ancestors $u_1$ and $u_2$ are \emph{induced} by the leaf $i$) and maximum total weight. We also consider \emph{special} HIA queries in which we are to find the heaviest ancestor of $v_{3-j}$ that is induced with $v_j$, for a specified $j \in \{1,2\}$. Gagie et al.~\cite{DBLP:conf/cccg/GagieGN13} provide several trade-offs for the space and query time of a data structure for answering HIA queries, some of which were recently improved by Abedin et al.~\cite{DBLP:conf/cpm/AbedinH0T18}. All of them are based on heavy-path decompositions $H(v_1,\mathcal{T}_1)$, $H(v_2,\mathcal{T}_2)$. In the following lemma we use the variant of the data structure from Section 2.2 in~\cite{DBLP:conf/cccg/GagieGN13}, substituting the data structure used from~\cite{DBLP:conf/compgeom/ChanLP11} to the one from~\cite{DBLP:conf/focs/AlstrupBR00} to obtain the following trade-off with a specified construction time. It can be readily verified that their technique answers special HIA queries within the same complexity.

\begin{lemma}[\cite{DBLP:conf/cccg/GagieGN13}]\label{lem:HIA}
HIA queries and special HIA queries over two weighted trees $\mathcal{T}_1$ and $\mathcal{T}_2$ of total size $\cO(n)$ can be answered in $\cO(\log^2 n)$ time, using a data structure of size $\cO(n\log^2 n)$ that can be constructed in $\cO(n\log^2 n)$ time.
\end{lemma}

Amir et al.~\cite{Amir2017} observed that the problem of computing an LCS after a single edit operation at position $i$ can be decomposed into two queries out of which we choose the one with the maximal answer: an occurrence of an LCS either avoids $i$ or it covers $i$. The former case can be precomputed. The latter can be reduced to HIA queries over suffix trees. It can be formalized by our \textsc{Three Substrings LCS} problem. The lemma below was implicitly shown in \cite{Amir2017} and \cite{DBLP:conf/cpm/AbedinH0T18}.


\threeSubstringsSpecial*
\begin{proof}
  We first construct $\mathcal{T}_1=\mathcal{T}(T^R)$ and $\mathcal{T}_2=\mathcal{T}(T)$ and the data structure for computing the loci of substrings (Corollary~\ref{cor:LAQ}). The leaf corresponding to prefix $T^{(i-1)}$ in $\mathcal{T}_1$ and to suffix $T_{(i)}$ in $\mathcal{T}_2$ are labeled with $i$. For the sake of HIA queries, we treat $\mathcal{T}_1$ and $\mathcal{T}_2$ as weighted trees over the set of explicit nodes. (Recall that each node has a string-depth and this is its weight.) Then we have:
  
  \begin{claim}
    Explicit nodes $u_1$ of $\mathcal{T}_1$ and $u_2$ of $\mathcal{T}_2$ are induced together if and only if $\mathcal{L}(u_1)^R \mathcal{L}(u_2)$ is a substring of~$T$.
  \end{claim}
  
  Assume we are to answer a \textsc{Three Substrings LCS} query for $U$, $V$, and $W=T$. Let $v_1$ be the locus of $U^R$ in $\mathcal{T}_1$ and $v_2$ be the locus of $V$ in $\mathcal{T}_2$. By the claim, if both $v_1$ and $v_2$ are explicit, then the problem reduces to an HIA query for $v_1$ and $v_2$. Otherwise, we ask an HIA query for the lowest explicit ancestors of $v_1$ and $v_2$ and special HIA queries for the closest explicit descendant of $v_j$ and $v_{3-j}$ for $j \in \{1,2\}$ such that $v_j$ is implicit.
\end{proof}

In Section~\ref{sec:internal} we show a solution to the general version of the \textsc{Three Substrings LCS} problem using a generalization of the HIA queries.

\subsection{Internal Queries for Special Substrings}\label{sec:int1}
We show how to answer internal LCS queries for a prefix or suffix of $S$ and a prefix or suffix of $T$ and for a substring of $S$ and the whole $T$.

In the solutions we use the formula:
\begin{equation}\label{eq:internal-LCS}
  |\LCS(S[a \dd b],T[c \dd d])|=\max_{\substack{i=a,\ldots,b,\\j=c,\ldots,d}} \,\{\min\{\lcp(S_{(i)},T_{(j)}),\,b-i+1,\,d-j+1\}\,\}.
\end{equation}
We also apply the following observation to create range maximum queries data structures over points constructed from explicit nodes of the GST $\mathcal{T}(S,T)$.

\begin{observation}\label{obs:lcp_formula}
  $\{\lcp(S_{(i)},T_{(j)})\,:\,i,j=1,\ldots,n\} \subseteq \{\mathcal{D}(v)\,:\,v\text{ is explicit in }\mathcal{T}(S,T)\}$.
\end{observation}

\prefixSuffix*
\begin{proof}
  For a node $v$ of $\mathcal{T}(S,T)$ and $U \in \{S,T\}$ we define:
  
  \vspace*{-0.3cm}
  \begin{align*}
    \minPref(v,U) &= \min\{i\,:\,\mathcal{L}(v)\text{ is a prefix of }U_{(i)}\},\\
    \maxPref(v,U) &= \max\{i\,:\,\mathcal{L}(v)\text{ is a prefix of }U_{(i)}\}.
  \end{align*}
  We assume that $\min \emptyset = \infty$ and $\max\emptyset = -\infty$. These values can be computed for all explicit nodes of $\mathcal{T}(S,T)$ in $\cO(n)$ time in a bottom-up traversal of the tree.

  We only consider computing $\LCS(S_{(a)},T_{(b)})$ and $\LCS(S^{(a)},T_{(b)})$ as the remaining cases can be solved by considering the reversed strings.

  In the first case formula \eqref{eq:internal-LCS} has an especially simple form:
  \[|\LCS(S_{(a)},T_{(b)})|=\max_{i \ge a,\,j \ge b} \lcp(S_{(i)},T_{(j)})\]
  which lets us use orthogonal range maximum queries to evaluate it. For each explicit node $v$ of $\mathcal{T}(S,T)$ with descendants from both $S$ and $T$ we create a point $(x,y)$ with weight $\mathcal{D}(v)$, where $x=\maxPref(v,S)$ and $y=\maxPref(v,T)$. By Observation~\ref{obs:lcp_formula}, the sought LCS length is the maximum weight of a point in the rectangle $[a,n] \times [b,n]$. This lets us also recover the LCS itself. The complexity follows from Lemma~\ref{lem:Ddrmq}.
  
  In the second case the formula \eqref{eq:internal-LCS} becomes:
  \[|\LCS(S^{(a)},T_{(b)})|=\max_{i \le a,\,j \ge b} \min(\lcp(S_{(i)},T_{(j)}),\,a-i+1).\]
  The result is computed in one of two steps depending on which of the two terms produces the minimum. First let us consider the case that $\lcp(S_{(i)},T_{(j)}) < a-i+1$. For each explicit node $v$ of $\mathcal{T}(S,T)$ with descendants from both $S$ and $T$ we create a point $(x,y)$ with weight $\mathcal{D}(v)$, where $x=\minPref(v,S)+\mathcal{D}(v)-1$ and $y=\maxPref(v,T)$. The answer $r_1$ is the maximum weight of a point in the rectangle $[1,a-1] \times [b,n]$.
  
  In the opposite case we can assume that the resulting internal LCS is a suffix of $S^{(a)}$ that does not occur earlier in $S$. For each explicit node $v$ of $\mathcal{T}(S,T)$ we create a point $(x,y)$ with weight $x'$, where $x'=\minPref(v,S)$, $x=x'+\mathcal{D}(v)-1$, and $y=\maxPref(v,T)$. Let $i$ be the minimum weight of a point in the rectangle $[a,n] \times [b,n]$. If $i \le a$, then we set $r_2=a-i+1$. Otherwise, we set $r_2=-\infty$.
  
  In both cases we use the 2d $\RMQ$ data structure of Lemma~\ref{lem:Ddrmq}. In the end, we return $\max(r_1,r_2)$ and the corresponding LCS.
\end{proof}

\substringString*
\begin{proof}
  We define $B[i]=\max_{j=1,\ldots,|T|}\, \{\lcp(S_{(i)},T_{(j)})\}$. The following fact was shown in \cite{Amir2017}. Here we give a proof for completeness.

\begin{claim}[\cite{Amir2017}]
  The values $B[i]$ for all $i=1,\ldots,|S|$ can be computed in $\cO(n)$ time.
\end{claim}
\begin{proof}
  For every explicit node $v$ of $\mathcal{T}(S,T)$ let us compute, as $\ell(v)$, the length of the longest common prefix of $\mathcal{L}(v)$ and any suffix of $T$. The values $\ell(v)$ are computed in a top-down manner. If $v$ has as a descendant a leaf from $T$, then clearly $\ell(v)=\mathcal{D}(v)$. Otherwise, we set $\ell(v)$ to the value computed for $v$'s parent. Finally, the values $B[i]$ can be read at the leaves of $\mathcal{T}(S,T)$.
\end{proof}

  The formula \eqref{eq:internal-LCS} can be written as:
$$|\LCS(S[a\dd b],T)| = \max_{a \le k \le b} \{\min(B[k],b-k+1)\}.$$

The function $f(k)=b-k+1$ is decreasing. We are thus interested in the smallest $k_0 \in [a,b]$ such that $B[k_0] \ge b-k_0+1$. If there is no such $k_0$, we set $k_0=b+1$. This lets us restate the previous formula as follows:
$$|\LCS(S[a\dd b],T)|=\max( \max_{a \le k < k_0} \{B[k]\},\, b-k_0+1).$$
Indeed, for $a \le k<k_0$ we know that $\min(B[k],b-k+1) = B[k]$, for $k=k_0$ we have $\min(B[k],b-k+1) = b-k_0+1$, and for $k_0 < k \le b$ we have
$\min(B[k],b-k+1) \le b-k+1 \le b-k_0+1$.

The final formula for LCS length can be evaluated in $\cO(1)$ time with a data structure for range maximum queries that can be constructed in linear time~\cite{Bender2000} on $B[1],\ldots,B[n]$, provided that $k_0$ is known. This lets us also recover the LCS itself.

\emph{Computation of $k_0$.}
The condition for $k_0$ can be stated equivalently as $B[k_0]+k_0 \ge b+1$. We create an auxiliary array $B'[i] = B[i]+i$. To find $k_0$, we need to find the smallest index $k \in [a,b]$ such that $B'[k] \ge b+1$. 
We can do this in time $\cO(\log n)$ by performing a binary search for $k$ in the range $[a,b]$ of $B'$ using $\cO(n)$-time preprocessing and $\cO(1)$-time range maximum queries~\cite{Bender2000}.
\end{proof}

\subsection{Three Substrings LCS Queries}\label{sec:int2}
We next show how to answer the general \textsc{Three Substrings LCS} queries.
We use \emph{extended} HIA queries that we define for two weighted trees $\mathcal{T}_1$ and $\mathcal{T}_2$ with the same set of $n$ leaves, numbered $1$ through $n$, as follows: given $1 \le a \le b \le |T|$, a node $v_1$ of $\mathcal{T}_1$, and a node $v_2$ of $\mathcal{T}_2$, return an ancestor $u_1$ of $v_1$ and an ancestor $u_2$ of $v_2$ such that:
\begin{enumerate}
\item $u_1$ and $u_2$ are induced by $i$;
\item $\mathcal{D}(u_j)=d_j$ for $j=1,2$;
\item $a \leq i-d_1$ and $i+d_2 \leq b+1$;
\item $d_1+d_2$ is maximal.
\end{enumerate}
We also consider \emph{special extended} HIA queries, in which the condition $u_1=v_1$ or the condition $u_2=v_2$ is imposed. Both extended and special extended HIA queries can be answered efficiently using multidimensional range maximum queries.

The motivation of the above definitions (extended and special extended HIA queries) becomes clearer in the proof of~\cref{lem:3-substring-LCS}. Intuitively, the additional hardness is due to the fact that we use this type of queries to answer \textsc{Three Substrings LCS} for an {\em arbitrary substring} $W=T[a\dd b]$ instead of $W=T$. To this end, we have extended the HIA queries and present a data structure to answer them efficiently. The proposed data structure is based on a non-trivial combination of heavy-path decompositions and multidimensional range maximum data structures. In Appendix~\ref{app:dimensions} we discuss how to reduce the number of dimensions from 6 to 4 (or even 3).

\begin{lemma}\label{lem:extended-HIA}
  Extended HIA queries and special extended HIA queries over two weighted trees $\mathcal{T}_1$ and $\mathcal{T}_2$ of total size $\cO(n)$ can be answered in $\Ohtilde(1)$ time after $\Ohtilde(n)$ time and space preprocessing.
\end{lemma}
\begin{proof}
  We defer answering special extended HIA queries until the end of the proof.
  Let us consider heavy paths in $\mathcal{T}_1$ and $\mathcal{T}_2$. Let us assign to each heavy path $\pi$ in $\mathcal{T}_j$, for $j=1,2$, a unique integer identifier of magnitude $\Oh(n)$ denoted by $\id(\pi)$. For a heavy path $\pi$ and $i \in \{1,\ldots,n\}$, by $d(\pi,i)$ we denote the depth of the lowest node of $\pi$ that has leaf $i$ in its subtree.

  We will create four collections of weighted points $\mathcal{P}^{I}$, $\mathcal{P}^{II}$, $\mathcal{P}^{III}$, $\mathcal{P}^{IV}$ in 6d. Let $i \in \{1,\ldots,n\}$. There are at most $\log n+1$ heavy paths on the path from leaf number $i$ to the root of $\mathcal{T}_1$ and $\mathcal{T}_2$. For each such {\em pair} of heavy paths, $\pi_1$ in $\mathcal{T}_1$ and $\pi_2$ in $\mathcal{T}_2$, we denote $d(\pi_j,i)$, for $j=1,2$, by $d_j$ and insert the point:
  \begin{itemize}
    \item $(\id(\pi_1),\id(\pi_2),d_1,d_2,i-d_1,i+d_2)$ to $\mathcal{P}^{I}$ with weight $d_1+d_2$;
    \item $(\id(\pi_1),\id(\pi_2),d_1,d_2,i,i+d_2)$ to $\mathcal{P}^{II}$ with weight $d_2$;
    \item $(\id(\pi_1),\id(\pi_2),d_1,d_2,i-d_1,i)$ to $\mathcal{P}^{III}$ with weight $d_1$;
    \item $(\id(\pi_1),\id(\pi_2),d_1,d_2,i,i)$ to $\mathcal{P}^{IV}$ with weight $0$.
  \end{itemize}
  Thus each collection contains $\cO(n \log^2 n)$ points. We perform preprocessing for range maximum queries on each of the collections by applying~\cref{lem:Ddrmq}.
  
  Assume that we are given an extended HIA query for $v_1$, $v_2$, $a$, and $b$ (inspect~\cref{fig:RMQ} for an illustration). We consider all the prefix fragments of heavy paths in $H(v_1,\mathcal{T}_1)$ and $H(v_2,\mathcal{T}_2)$. For $j=1,2$, let $\pi'_j$ be a prefix fragment of heavy path $\pi_j$ in $H(v_j,\mathcal{T}_j)$, connecting node $x_j$ with its descendant $y_j$. 
  Now suppose that for some $i$, $d(\pi_j,i)>\mathcal{D}(y_j)$;
  we essentially want to reassign $i$ to $y_j$ which is the deepest ancestor of $i$ in $\pi'_j$.
  To this end, we define intervals $I_j=[\mathcal{D}(x_j),\,\mathcal{D}(y_j)-1]$ and $I_j^\infty=[\mathcal{D}(y_j),\infty)$.
  
  For each of the $\cO(\log^2 n)$ pairs of prefix fragments $\pi'_1$ and $\pi'_2$ in the decompositions of the root-to-$v_1$ and root-to-$v_2$ paths, respectively, we ask four range maximum queries, to obtain the following values: 
  \begin{enumerate}
    \item $\RMQ_{\mathcal{P}^{I}}(\id(\pi_1),\,\id(\pi_2),\,I_1,\,I_2,\,[a,\infty),\,(-\infty,b+1])$ that corresponds to finding a pair of induced nodes $u_1 \in \pi'_1 \smallsetminus \{y_1\}$ and $u_2 \in \pi'_2 \smallsetminus \{y_2\}$;
    \item $\mathcal{D}(y_1) + \RMQ_{\mathcal{P}^{II}}(\id(\pi_1),\,\id(\pi_2),\,I_1^\infty,\,I_2,\,[a+\mathcal{D}(y_1),\infty),\,(-\infty,b+1])$ that corresponds to finding a pair of induced nodes $u_1=y_1$ and $u_2 \in \pi'_2 \smallsetminus \{y_2\}$;
    \item $\mathcal{D}(y_2) + \RMQ_{\mathcal{P}^{III}}(\id(\pi_1),\,\id(\pi_2),\,I_1,\,I_2^\infty,\,[a,\infty),\,(-\infty,b+1-\mathcal{D}(y_2)])$ that corresponds to finding a pair of induced nodes $u_1 \in \pi'_1 \smallsetminus \{y_1\}$ and $u_2=y_2$;
    \item $\mathcal{D}(y_1) + \mathcal{D}(y_2) + \RMQ_{\mathcal{P}^{IV}}(\id(\pi_1),\,\id(\pi_2),\,I_1^\infty,\,I_2^\infty,\,[a+\mathcal{D}(y_1),\infty),\,(-\infty,b+1-\mathcal{D}(y_2 )])$ that corresponds to checking if $y_1$ and $y_2$ are induced.
  \end{enumerate}
  
\begin{figure}[htpb]
\begin{center}
\begin{tikzpicture}[scale=0.6]
  \begin{scope}
    \draw[thick] (0,0) -- node[below=0.2cm] {$\mathcal{T}_1$} (10,0) -- (5,10) -- (0,0);
    \draw (1.5,0) -- (4.5,7) -- (7.5,0);
    \draw (3.8,0) -- (5,5) -- (6.2,0);
    \draw[thick] (4.5,7) -- (4.5,6) -- (5,5) -- (4.9,3.7) -- (5.2,3);
    \draw (4.5,7) node[above] {$x_1$};
    \draw (5.05,5) node[left] {$y_1$};
    \draw (5.3,5) node[right] {$\pi_1$};
    \draw (2.65,0) node[above] {$I_1$};
    \draw (5.1,0) node[above] {$I_1^\infty$};
    \draw (6.85,0) node[above] {$I_1$};
    \draw[snake=brace] (2,5) -- node[left] {$\pi'_1$} (2,7);
    \filldraw (4.5,7) circle (0.1cm);
    \filldraw[color=red] (4.5,6) circle (0.1cm);
    \filldraw[color=blue] (5,5) circle (0.1cm);
    \filldraw (4.9,3.7) circle (0.1cm);
    \draw[densely dashed] (4.5,6) -- (2.9,1.9);
    \draw (2.9,1.9) node[below] {$i$};
    \draw[densely dashed] (4.9,3.7) -- (4.5,2.0333);
    \draw (4.5,2.0333) node[below] {$i'$};
  \end{scope}
  \begin{scope}[xshift=12cm]
    \draw[thick] (0,0) -- node[below=0.2cm] {$\mathcal{T}_2$} (10,0) -- (5,10) -- (0,0);
    \draw (2,0) -- (5.5,8) -- (9,0);
    \draw (4.2,0) -- (5.7,4.8) -- (7.2,0);
    \draw[thick] (5.5,8) -- (5.7,7) -- (5.3,6) -- (5.7,4.8) -- (5.7,3);
    \draw (5.4,8) node[above] {$x_2$};
    \draw (5.7,4.8) node[left] {$y_2$};
    \draw (5.3,5.76) node[right] {$\pi_2$};
    \draw (3.1,0) node[above] {$I_2$};
    \draw (5.7,0) node[above] {$I_2^\infty$};
    \draw (8,0) node[above] {$I_2$};
    \draw[snake=brace] (8,8) -- node[right] {$\pi'_2$} (8,4.8);
    \filldraw (5.5,8) circle (0.1cm);
    \filldraw[color=blue] (5.7,7) circle (0.1cm);
    \filldraw[color=red] (5.3,6) circle (0.1cm);
    \filldraw (5.7,4.8) circle (0.1cm);
    \draw[densely dashed] (5.7,7) -- (7.9325,1.9);
    \draw (7.9325,1.9) node[below] {$i'$};
    \draw[densely dashed] (5.3,6) -- (3.50625,1.9);
    \draw (3.50625,1.9) node[below] {$i$};
  \end{scope}
\end{tikzpicture}
\end{center}
\caption{Illustration of the notations used to handle a pair of prefix fragments $\pi'_1$ and $\pi'_2$.
The descendant leaves of $x_j$ are implicitly partitioned at query time, by employing intervals $I_j$ and $I_j^\infty$, according to whether their deepest ancestor in $\pi_j$ is a strict ancestor of $y_j$ or not.
For example, the pair of red nodes, induced by $i$, will be considered by the \RMQ{} of type 1, while the pair of blue nodes, induced by $i'$, by the \RMQ{} of type 2, assuming that the last two constraints of the respective $\RMQ$ are satisfied.}
\label{fig:RMQ}
\end{figure}
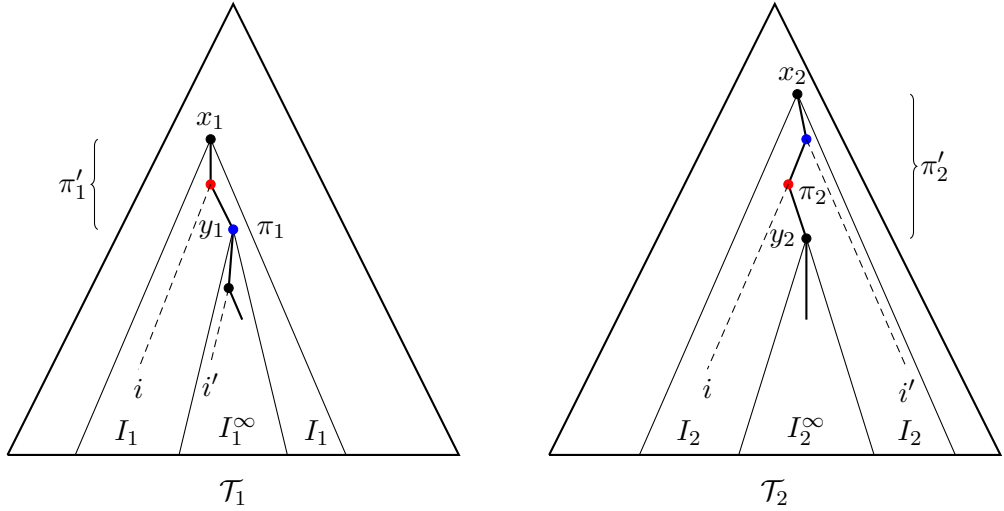
  
  If an RMQ concerns an empty set of points, it is assumed to return $-\infty$. We return the point that yielded the maximal value.
    
  Hence, an extended HIA query reduces to $\Oh(\log^2 n)$ range maximum queries in collections of points in 6d of size $\Oh(n \log^2 n)$. A special extended HIA query can be answered in a simpler way by asking just $\Oh(\log n)$ RMQs, where one of the prefix fragments $\pi'_1$, $\pi'_2$ reduces always to a single node. The statement follows.
  \end{proof}
  
The proof of the following lemma is very similar to the proof of Lemma~\ref{lem:3-substring-LCS-special}; we simply use extended HIA queries instead of the regular ones.

\threeSubstring*
\begin{proof}
  Let us generalize the proof of Lemma~\ref{lem:3-substring-LCS-special}. We construct $\mathcal{T}_1=\mathcal{T}(T^R)$, $\mathcal{T}_2=\mathcal{T}(T)$, and the data structure for computing the loci of substrings (Corollary~\ref{cor:LAQ}). The leaf corresponding to prefix $T^{(i-1)}$ in $\mathcal{T}_1$ and to suffix $T_{(i)}$ in $\mathcal{T}_2$ are labeled with $i$. Let $W=T[a \dd b]$ be an occurrence of $W$ in $T$. If we treat $\mathcal{T}_1$ and $\mathcal{T}_2$ as weighted trees over the set of explicit nodes, then we have:
  
  \begin{claim}
    Let $u_1$ and $u_2$ be explicit nodes of $\mathcal{T}_1$ and $\mathcal{T}_2$, respectively, and $d_1=\mathcal{D}(u_1)$, $d_2=\mathcal{D}(u_2)$. Then $\mathcal{L}(u_1)^R \mathcal{L}(u_2)$ is a substring of $W$ if and only if $u_1$ and $u_2$ are induced by $i$ such that $a \le i-d_1$ and $i+d_2-1 \le b$.
  \end{claim}
  \begin{proof}
    $(\Rightarrow)$ The string $\mathcal{L}(u_1)^R \mathcal{L}(u_2)$ is a substring of $W=T[a \dd b]$, so there exists an index $i \in [a \dd b]$ such that $\mathcal{L}(u_1)^R$ is a suffix of $T[a \dd i-1]$ and $\mathcal{L}(u_2)$ is a prefix of $T[i \dd b]$. This implies that $u_1$ and $u_2$ are induced by $i$ and that:
    $$
      d_1 = |\mathcal{L}(u_1)^R| \le |T[a \dd i-1]|=i-a \quad\mbox{and}\quad
      d_2 = |\mathcal{L}(u_2)| \le |T[i \dd b]|=b-i+1.
    $$
  
  $(\Leftarrow)$ If $u_1$ and $u_2$ are induced by $i$, then $\mathcal{L}(u_1)^R$ occurs as a suffix of $T^{(i-1)}$ and $\mathcal{L}(u_2)$ occurs as a prefix of $T_{(i)}$. By the inequalities $a \le i-d_1$ and $i+d_2-1 \le b$, $\mathcal{L}(u_1)^R \mathcal{L}(u_2)$ is a substring of $T[a \dd b]=W$.
  \end{proof}
  
  Assume we are given a \textsc{Three Substrings LCS} query for $U$, $V$, and $W=T[a \dd b]$. Let $v_1$ be the locus of $U^R$ in $\mathcal{T}_1$ and $v_2$ be the locus of $V$ in $\mathcal{T}_2$. By the claim, if both $v_1$ and $v_2$ are explicit, then the problem reduces to an extended HIA query for $v_1$ and $v_2$. Otherwise, we ask an extended HIA query for the lowest explicit ancestors of $v_1$ and $v_2$ and special extended HIA queries for the closest explicit descendant of $v_j$ and $v_{3-j}$ for $j \in \{1,2\}$ such that $v_j$ is implicit.
\end{proof}

\subsection{Decremental Bounded-Length LCS}\label{sec:decr}
In the following lemma we extensively use a data structure $\mathcal{B}$ that allows insertion of intervals (operation $\iinsert_{\mathcal{B}}([p,q])$) and asking queries for the number of intervals covering a given point (operation $\ccount_{\mathcal{B}}(r)$). This data structure, sometimes called an \emph{interval tree}, can be implemented in linear space and logarithmic query time using a balanced BST as follows. 

For each interval $[p,q]$ insertion we add to the BST $(p,1)$ and $(q+1,-1)$, where the rank of a pair $(a,b)$ is $a$; if an element with rank $p$ (resp.~$q+1$) is already in the BST we instead increment (resp.~decrement) the second integer of the pair by $1$. We also store in each BST node $v$ the sum of the second elements of the pairs in the subtree rooted at $v$. Hence, $\iinsert_{\mathcal{B}}([p,q])$ takes $\cO(\log n)$ time. Given a query for a point $r$, we search for it in the BST and by exploiting the subtree sums stored at the nodes of the tree we can implement $\ccount_{\mathcal{B}}(r)$ in $\cO(\log n)$ time as well.

\boundedLength*
\begin{proof} 
We first compute an array $A[0 \dd d]$ such that $A[i]$ stores the number of distinct strings of length $i$ that have occurrences in both $S$ and $T$.
\begin{claim}
$A$ can be initialized in $\cO(n)$ time.
\end{claim}
\begin{proof}
We start by creating a copy $\mathcal{T}_p$ of $\mathcal{T}(S,T)$ and pruning it, only keeping the nodes that have descendants from both $S$ and $T$. The following information on $\mathcal{T}_p$, which can be computed in $\cO(n)$ time since $\mathcal{T}_p$ has $\cO(n)$ nodes, suffices for computing array $A$:
\begin{itemize}
\item $L[i]$: number of leaves at string-depth $i$;
\item $C[i]$: a multiset consisting of the out-degrees of explicit nodes at string-depth $i$.
\end{itemize}
We then set $A[0]=1$ and obtain $A[i]$ by the following formula: $A[i]=A[i-1]-L(i-1)+\sum_{j \in C[i-1]}(j-1)$.
\end{proof}
Recall that if $S[p]$ is replaced by a $\#$, then every substring $S[a \dd b]$ for $a \le p \le b$ is called \emph{invalidated} and the same name is used for substring occurrences that are invalidated by character replacements in $T$. We say that a substring of $S$ (resp.~$T$) is $S$-invalidated (resp.~$T$-invalidated) if all of its occurrences in $S$ (resp.~$T$) are invalidated. Whenever a substring of length at most $d$ has been $S$-invalidated and not $T$-invalidated, we need to update the array $A$ accordingly. This will be implemented with the aid of an interval tree $\mathcal{B}$, such that insertion of an interval $[p,q]$ to $\mathcal{B}$ corresponds to decrementation of $A[p],\ldots,A[q]$.

For each explicit node $v$ of $\mathcal{T}(S,T)$ we store $C_S(v)$ and $C_T(v)$, denoting the number of occurrences of $\mathcal{L}(v)$ in $S$ and $T$, respectively; we can compute these numbers in a depth first traversal of $\mathcal{T}(S,T)$.

We then compute a heavy-path decomposition of $\mathcal{T}(S,T)$. We consider the implicit nodes of a light edge $(u,w)$ to be part of the heavy path rooted at $w$. For each heavy path~$\pi$, as $I_S(\pi)$ ($I_T(\pi)$) we store the depth of the lowest node of $\pi$ that has been $S$-invalidated (resp.~$T$-invalidated). This is due to the following simple observation.

\begin{observation}\label{obs:trivial}
A substring $U$ is not $S$-invalidated ($T$-invalidated) later than any of its prefixes.
\end{observation}

Moreover, for the heavy path $\pi$ we store two interval trees, $\mathcal{B}_S(\pi)$ and $\mathcal{B}_T(\pi)$.

\textit{Implementation of character replacement.}
Let us now describe the implementation of a replacement of $S[p]$ by $\#$; replacements in $T$ are handled symmetrically. For every $a=p,\ldots,\min\{p-d,\pp_{S'}(p)\}+1$, this replacement invalidates substrings $S[a \dd b]$ for $b=p,\ldots,\min\{a+d,\nn_{S'}(p)\}-1$. The set of nodes of $\mathcal{T}(S,T)$ that represent those substrings for a fixed~$a$ forms a path. Its endpoints can be identified in $\Oh(\log \log n)$ time due to Corollary~\ref{cor:LAQ}. This path can be decomposed into fragments of $\Oh(\log n)$ heavy paths; all but possibly the topmost fragment of the decomposition are prefix fragments (including the top light edge). For a prefix fragment of heavy path $\pi$ with topmost node $u$ and bottommost node $w$, the interval $[\mathcal{D}(u),\mathcal{D}(w)]$ is inserted to $\mathcal{B}_S(\pi)$.

Let $v$ be an (explicit of implicit) node of $\pi$ of string-depth at most $d$ and $v'$ be its closest explicit descendant. At this point, the number of non-invalidated occurrences of the substring $\mathcal{L}(v)$ in $S$ equals:
\begin{equation}\label{eq:num_occ}
  C_S(v')-\ccount_{\mathcal{B}_S(\pi)}(\mathcal{D}(v)).
\end{equation}
Hence, $\mathcal{L}(v)$ is $S$-invalidated if and only if this number is 0.

At this point $I_S(\pi)$ needs to be updated. By Observation~\ref{obs:trivial}, this can be done using binary search, with each step of the binary search implemented by checking if the formula~\eqref{eq:num_occ} for the tested node $v$ yields a positive value. Let the new value of $I_S(\pi)$ be $I'_S(\pi)$. Then all the elements of the array $A$ in the interval $[I'_S(\pi),\min \{ I_S(\pi),I_T(\pi)\}-1]$ need to be decremented. Hence, this interval is inserted into $\mathcal{B}$, and we set $I_S(\pi):=I'_S(\pi)$.

To find the length of the longest common substring of $S$ and $T$ of length at most $d$ after the character replacement, it suffices to query the interval tree $\mathcal{B}$ for each $r=1,\ldots,d$. This allows us to recompute the array $A$ and clear the interval tree $\mathcal{B}$. 

\textit{Retrieving the substrings.}
To retrieve a pair of positions where the actual substring occurs in $S$ and $T$, we maintain an array $R$ of size $d$ of linked lists, where we store representatives of explicit nodes of $\mathcal{T}(S,T)$ that are topmost nodes of some heavy path. A representative of a node $v$ and the actual node $v$ store pointers to each other. A topmost node $v$ is stored in $R[i]$, where $i=\min \{ I_S(\pi),I_T(\pi)\}-1$. Suppose that after some operation we have $m=\max\{i|A[i] \neq 0\}$. We can then obtain the topmost node of a heavy path whose unique leaf's prefix of length $m$ occurs in both $S$ and $T$ employing $R[m]$. We can then find a pair of positions where this string occurs in $S$ and $T$ as described below. 

For each heavy path $\pi$ we construct and maintain two balanced BSTs $E_S(\pi)$ and $E_T(\pi)$. $E_X(\pi)$, for $X \in \{S,T\}$, stores the leaves of the subtree rooted at the topmost node of $\pi$. Let the unique leaf in $\pi$ be $u_{\pi}$, corresponding to $X_{(t)}$. Throughout the computation, we maintain that a leaf $w \in E_X(\pi)$ corresponding to $X_{(j)}$ has rank $\min\{\lcp(X_{(t)},X_{(j)}),d,\nn_{X'}(j)-j\}$ in $E_X(\pi)$.
We can do this by initially setting the rank of $w$ in $E_X(\pi)$ to be $\min\{\lcp(X_{(t)},X_{(j)}),d\}$ and then updating it each time a prefix of $X[j \dd j+d-1]$ gets invalidated. During the query we just find the maximum in each of $E_S(\pi)$ and $E_T(\pi)$.

\textit{Complexity.}
Each character replacement in $S$ or $T$ implies $\Oh(d\log n)$ insertions into the interval trees $\mathcal{B}_S(\pi)$ and $\mathcal{B}_T(\pi)$. Hence, the time complexity of this operation is $\Ohtilde(d)$ and the total size of the interval trees after $k$ operations is $\Ohtilde(kd)$. Storing $\mathcal{T}(S,T)$ requires $\Oh(n)$ space. The binary search used to update $I_S(\pi)$ or $I_T(\pi)$ takes $\Ohtilde(1)$ time. A character replacement inserts $\Oh(d \log n)$ intervals into the interval tree associated with the array $A$. Finally, the balanced binary search trees $E_S(\pi)$ and $E_T(\pi)$ have total size $\cO(n\log n)$ and can be implemented in $\cO(n \log^2 n)$ time and $\cO(n\log n)$ space in total. Each update or query on them can be performed in $\cO(\log n)$ time --- note that we assume $\cO(1)$-time access to the leaves of $\mathcal{T}(S,T)$. The complexity follows.
\end{proof}  

\section{Fully Dynamic Longest Repeat}\label{sec:lr}

In the longest repeat problem we are given as input a string $S$ of length $n$ and we are to report a longest substring that occurs at least twice in $S$. This can be done in $\Oh(n)$ time and space~\cite{weiner1973linear}. In the fully dynamic longest repeat problem we are given as input a string $S$, which we are to maintain under subsequent edit operations, so that after each operation we can efficiently return a longest repeat in $S$. The application of our techniques for the LCS problem is quite straightforward, which is not surprising given the connection between the two problems. In what follows we briefly discuss the modifications required to the algorithms presented for the different subcases which we have decomposed the LCS problem into; we decompose the longest repeat problem in an analogous manner.



\subsection{Decremental Longest Repeat}
\cref{lem:aux_approx} covers the case that the LCS is short. It essentially provides an efficient way to maintain counters on the distinct substrings of $S$ and $T$ up to a specified length that occur in both strings. We then just read the counters and identify the longest such substring that has at least one occurrence at each of the two strings. We can instead maintain this information for one string, keeping counters on the number of distinct substrings that have at least two occurrences for each length.

\cref{lem:problem} finds long enough substrings that are common in both strings and are hence guaranteed to be anchored in a pair of positions in the difference cover.
Let $A$ be this difference cover. 
Recall that we construct a compact trie $\T(\F)$ of a family of strings $\F$, defined in terms of the input strings, the difference cover and the updated positions.
Then, given $\T(\F)$ and sets $\P,\Q\subseteq \F^2$ we efficiently compute $\max\{\lcp(P_1,Q_1)+\lcp(P_2,Q_2) : (P_1,P_2)\in \P \text{ and }(Q_1,Q_2)\in \Q\}$ by using an efficient solution to the \textsc{Two String Families LCP} problem. Sets $\P$ and $\Q$ essentially allow us to distinguish between substrings of $S'$ and $T'$ in $\mathcal{F}$.

We adapt this for the longest repeat problem with an extra --- yet avoidable --- $\log n$ factor in the complexities. We build $\cO(\log n)$ copies of the respective tree $\mathcal{T}(\mathcal{F})$, where in the $j$th copy, $0\leq j\leq \log n$, we set:

  \begin{align*}
    \P &= \{\,(S[\pp_{S'}(i-1)+1 \dd i-1])^R,\,S[i \dd \nn_{S'}(i)-1])\,:\,i \in A\,,\,\lfloor i/2^j \rfloor =2m\,,\,m \in \mathbb{Z}\,\}\\
    \Q &= \{\,(S[\pp_{S'}(i-1)+1 \dd i-1])^R,\,S[i \dd \nn_{S'}(i)-1])\,:\,i \in A\,,\,\lfloor i/2^j \rfloor =2m+1\,,\,m \in \mathbb{Z}\,\}.
  \end{align*}

Note that since $\P$ and $\Q$ are disjoint, in none of the tries do we align a fragment of $S'$ with itself. It is also easy to see that any two positions $r<t$ in $A$ will be in different sets at one of the copies of $\mathcal{T}(\mathcal{F})$. 

\subsection{Cross-Substring Queries}
Our algorithms for answering cross-substring queries for the LCS problem can be adapted effortlessly. One can think of this as creating two copies of $S'$ and applying the LCS algorithm as follows. The algorithm of~\cref{sec:fully.2} can be applied as is, since it guarantees that the starting position of the returned pair of substrings will not be the same.
As for the algorithm of~\cref{sec:fully.3}, one has only to disallow the aligned prefix-suffix to be starting at the same position at the two copies of $S'$; this can be done straightforwardly.

\subsection{Round-Up}
The discussion of the two subsections leads to the following result.

\begin{theorem}
  Fully dynamic longest repeat queries on a string of length up to $n$ can be answered in $\Ohtilde(n^{2/3})$ time per query, using $\Ohtilde(n)$ space, after $\Ohtilde(n)$-time preprocessing.
\end{theorem}

\section{Fully Dynamic Longest Palindrome Substring}\label{sec:pals}
A \emph{palindrome} is a string $U$ such that $U^R=U$. For a string $S$, by $\LSPal(S)$ let us denote a longest substring of $S$ that is a palindrome. We show a solution to the following auxiliary problem.

{\defDSproblem{$k$\textsc{-Substring LSPal}}{A string $S$ of length $n$}
{$\LSPal(S')$ where $S'=F_1 \ldots F_k$ is a $k$-substring of $S$}

The \emph{center} of a palindrome substring $S[i \dd j]$ is $\frac{i+j}{2}$. A palindrome substring $S[i \dd j]$ of $S$ is called \emph{maximal} if $i=1$, or $j=n$, or $S[i-1 \dd j+1]$ is not a palindrome. For short, we call maximal palindrome substrings MPSs. By $\M(S)$ we denote the set of all MPSs of $S$. For each integer or half-integer center between 1 and $n$ there is exactly one MPS with this center, so $|\M(S)| = 2n-1$. The set $\M(S)$ can be computed in $\Oh(n)$ time using Manacher's algorithm \cite{DBLP:journals/jacm/Manacher75} or $\LCE{}$ queries~\cite{Gusfield:1997:AST:262228}.

\subsection{Internal Queries}
In an internal \LSPal query we are to compute the longest palindrome substring of a given substring of $S$. In the following lemma we show that such queries can be reduced to 2d range maximum queries.
\begin{lemma}\label{lem:pal_internal}
  Let $S$ be a string of length $n$. After $\Oh(n\log^2 n)$-time and $\Oh(n \log n)$-space preprocessing, one can compute the \LSPal of a substring of $S$ and the longest prefix/suffix palindrome of a substring of $S$ in $\Oh(\log n)$ time.
\end{lemma}
\begin{proof}
  $\LSPal(S[i \dd j])$ is a substring of an MPS of $S$ with the same center. We consider two cases depending on whether the \LSPal is equal to the MPS.
  
  If this is the case, the MPS is a substring of $S[i \dd j]$. We create a 2d grid and for each $S[a \dd b] \in \M(S)$ we create a point $(a,b)$ with weight $b-a+1$. The sought MPS can be found by an RMQ for the rectangle $[i,\infty) \times (-\infty,j]$.
  
  In the opposite case, $\LSPal(S[i \dd j])$ is a prefix or a suffix of $S[i \dd j]$. We consider the case that it is a prefix of $S[i \dd j]$; the other case is symmetric.
  The longest prefix palindrome of $S[i \dd j]$ can be obtained by trimming to the interval of positions $[i \dd j]$, the MPS with the rightmost center, among the ones with starting position smaller than $i$ and center at most $(i+j)/2$.
  To this end, we create a new 2d grid. For each $S[a \dd b] \in \M(S)$, we create a point $(a,a+b)$ with weight $a+b$. The answer to an RMQ on the rectangle $(\infty,i] \times (-\infty,i+j]$ is twice the center of the desired MPS of $S$.
  
  In either case, we use Lemma~\ref{lem:Ddrmq} to answer a 2d RMQ; the complexity follows.
\end{proof}

If the answer to $k$\textsc{-Substring \LSPal} contains none of the changed positions $a_1,\ldots,a_k$, it can be found by asking $k+1$ internal \LSPal queries.

\subsection{Cross-Substring Queries}\label{sss:cont}
Let us assume that the boundary between $F_{i-1}$ and $F_i$ is the closest one to the center of \LSPal. In what follows, we consider the case that it lies to the left of the center. Then, the palindrome cut to the positions in $F_i$ a prefix palindrome of $F_i$. The opposite case is symmetric. In total, this gives rise to $2k$ cases that need to be checked.

The structure of palindromes being prefixes of a string has been well studied. It is known that a palindrome being a prefix of another palindrome is also its border, and a border of a palindrome is always a palindrome; see~\cite{DBLP:journals/jda/FiciGKK14}. Hence, the palindromes being prefixes of $F_i$ are the borders of the longest such palindrome, further denoted by $U_0$. We are interested in the borders of $U_0$.

The palindrome $U_0$ can be computed by Lemma~\ref{lem:pal_internal}. By Lemma~\ref{lem:borders}, the set of border lengths of $U_0$ can be divided into $\Oh(\log n)$ arithmetic sequences. Each such sequence will be treated separately. Assume first that it contains at most two elements. Let $f_i$ denote the starting position of $F_i$ in $S'$, for $i=1,\ldots,k$. Then for each element $u$ representing a palindrome $U$, the longest palindrome having the same center as $U$ in $S'$ has the length
\[u+2\cdot\lcp(S'_{(f_i+u)},((S')^{(f_i-1)})^R).\]
This \LCE{} query can be answered in $\Oh(\log n)$ time using~\cref{lem:lce_rand} for $S\#S^R$.

Now assume that the arithmetic sequence has more than two elements. Let $p$ be the difference of the arithmetic sequence, $\ell$ be its length, $u$ be its maximum element, and
\[X=S'_{(f_i+u)},\quad Y=((S')^{(f_i-1)})^R,\quad P=S'[f_i+u-p+2 \dd f_i+u-1].\]
Then the longest palindrome having the same center as an element of this sequence has length
\[2\cdot \max_{w=0}^{\ell-1} \{\,\lcp(P^wX,Y)+\tfrac12(u-wp)\,\}.\]
By Observation~\ref{obs:Pinfty}, this formula can be evaluated in $\Oh(\log^2 n)$ time.

Over all arithmetic sequences, we obtain $\Oh(\log^3 n)$ query time.

\subsection{Round-Up}
The results of the two subsections can be combined into an algorithm for $k$\textsc{-Substring LSPal}.

\begin{lemma}
  $k$\textsc{-Substring LSPal} queries can be answered in $\Oh(k \log^3 n)$ time after $\Oh(n\log^2 n)$-time and $\Oh(n \log n)$-space preprocessing.
\end{lemma}

Using the general scheme, we obtain a solution to fully dynamic longest palindrome substring problem.

\begin{theorem}
  Fully dynamic longest palindrome substring queries on a string of length up to $n$ can be answered in $\cO(\sqrt{n}\log^{2.5} n)$ time per query, using $\cO(n \log n)$ space, after $\cO(n\log^2 n)$-time preprocessing.
\end{theorem}

\section{Fully Dynamic Longest Lyndon Substring}\label{sec:lyndon}
\newcommand{\LF}{\mathit{LF}}
\newcommand{\RLF}{\mathit{LFR}}
\newcommand{\LTree}{\mathit{LTree}}
\newcommand{\rc}{\mathit{rc}}

Let us recall that string $S$ is smaller in the lexicographic order than string $T$, written as $S<T$, if $S$ is a proper prefix of $T$ or $S^{(i)}=T^{(i)}$ and $S[i+1]<T[i+1]$. Let us also recall that a string $S$ is a Lyndon string if and only if $S<S_{(j)}$ for all $j=2,\dots,|S|$ and that a Lyndon factorization of a string $S$, denoted as $\LF_S$, is the unique way of writing $S$ as $L_1 \ldots L_p$ where $L_1,\ldots,L_p$ are Lyndon strings (called \emph{factors}) that satisfy $L_1 \ge L_2 \ge \dots \ge L_p$ (if $S$ is a Lyndon string, then its Lyndon factorization is composed of $S$ itself).

The following fact was shown as Lemma~3 in Urabe et al.~\cite{DBLP:conf/cpm/UrabeNIBT18}.

\begin{lemma}[\cite{DBLP:conf/cpm/UrabeNIBT18}]
  The longest Lyndon substring of a string is the longest factor in its Lyndon factorization.
\end{lemma}

By a \emph{representation} of a Lyndon factorization of $S$, denoted as $\RLF_S$, we mean a data structure that supports the following queries on $\LF_S$ in $\Ohtilde(1)$ time:
\begin{itemize}
  \item \texttt{longest}: computing the longest factor in $\LF_S$
  \item \texttt{count}: computing the number of factors in $\LF_S$
  \item \texttt{select}$(i)$: computing the $i$th factor in $\LF_S$.
\end{itemize}
More precisely, our representation will support the first two types of queries in $\Oh(1)$ time and the third type in $\Oh(\log^2 n)$ time. In this section we tackle the following problem.

{\defDSproblem{$k$\textsc{-Substring Lyndon Factorization}}{A string $S$ of length $n$}
{$\RLF_{S'}$ where $S'$ is a $k$-substring of $S$}

Let us start with the definition of a Lyndon tree of a Lyndon string $W$~\cite{DBLP:journals/jct/Barcelo90}. If $W$ is not a Lyndon string, the tree is constructed for $\$W$, where $\$$ is a special character that is smaller than all the characters in $W$. An example of a Lyndon tree can be found in Fig.~\ref{fig:Ltree} below.

\begin{definition}
  The \emph{standard factorization} of a Lyndon string $S$ is $(U,V)$ if $S=UV$ and $V$ is the longest proper suffix of $S$ that is a Lyndon string. In this case, both $U$ and $V$ are Lyndon strings. The \emph{Lyndon tree} of a Lyndon string $S$ is the full binary tree defined by recursive standard factorization of $S$. More precisely, $S$ is the root; if $|S|>1$, its left child is the root of the Lyndon tree of $U$, and its right child is the root of the Lyndon tree of $V$.
\end{definition}

\subsection{Internal Queries}
If $v$ is a node of a binary tree, $u$ is its ancestor, $u \ne v$, and $w$ is the right child of $u$ that is not an ancestor of $v$, then we call $w$ a \emph{right uncle} of $v$. By $U(v)$ we denote the list of all right uncles of $v$ in bottom-up order. The following fact was shown in~\cite{DBLP:conf/cpm/UrabeNIBT18}.

\begin{lemma}[Lemma~12 in \cite{DBLP:conf/cpm/UrabeNIBT18}]\label{lem:sufLF}
  $\LF_{S_{(j)}}=U(v)$, where $v$ is the leaf of $\LTree(S)$ that corresponds to $S[j-1]$.
\end{lemma}

We say that $S$ is a \emph{pre-Lyndon string} if it is a prefix of a Lyndon string. The following lemma is a well-known property of Lyndon strings; see Lemma~10 in~\cite{DBLP:conf/cpm/UrabeNIBT18} or the book~\cite{Knuth:2005:ACP:1121684}.
\begin{lemma}[\cite{Knuth:2005:ACP:1121684,DBLP:conf/cpm/UrabeNIBT18}]\label{lem:prefLF}
  If $S$ is a pre-Lyndon string, then there exists a unique Lyndon string $X$ such that $S=X^kX'$ where $X'$ is a proper prefix of $X$ and $k \ge 1$. Moreover, $\LF_S=\underbrace{X,X,\ldots,X}_{k\text{ times}},\LF_{X'}$.
\end{lemma}
In this case $\LF_S$ can also be represented more efficiently in a \emph{compact} form  where the first part of the representation is simply written as $(X)^k$. The string $X$ from the lemma also satisfies the following property.
\begin{observation}\label{obs:prefLF}
  $|X|$ is the shortest period of $S$.
\end{observation}
\begin{proof}
  Clearly $|X|$ is a period of $S$. If it was not the shortest period, then $X$ would have a non-trivial period, hence a proper non-empty border. This contradicts the definition of a Lyndon string.
\end{proof}

The above properties are sufficient to determine Lyndon factorizations of prefixes and suffixes of a string $S$ according to~\cite{DBLP:conf/cpm/UrabeNIBT18} as follows:
\begin{itemize}
  \item To compute $\LF_{S^{(j)}}$, take the minimal number of prefix factors in $\LF_S$ that cover $S^{(j)}$, trim the last of these factors accordingly, compute the Lyndon factorization of the trimmed factor using Lemma~\ref{lem:prefLF}, and append it to the previous factors.
  \item To compute $\LF_{S_{(j)}}$, take the list of right uncles of the leaf $S[j-1]$ as shown in Lemma~\ref{lem:sufLF}.
\end{itemize}
We are now ready to describe $\LF_{S[i \dd j]}$. Let $v_1$ and $v_2$ be the leaves of $\LTree(S)$ that correspond to $S[i-1]$ and $S[j]$, respectively, $w$ be their lowest common ancestor (LCA), and $u$ be its right child. Then $\LF_{S[i \dd j]}$ is determined by taking the list of right uncles $U(v_1)$ up to $u$, trimming the factor in $u$ up to position $j$, and computing the Lyndon factorization of the trimmed factor according to Lemma~\ref{lem:prefLF}.

\begin{example}
We consider the Lyndon string $S=\texttt{\$eaccbcebcdbce}$, whose Lyndon tree is presented in~\cref{fig:Ltree}.
Let us suppose that we want to compute the Lyndon factorization of the substring $S[4 \dd 13]=\texttt{ccbcebcdbc}$.
We first find the LCA of the leaves representing $S[3]$ and $S[13]$. The path from this LCA to the leaf representing $S[3]$ is shown in blue.
The Lyndon factorization of $S[4 \dd 14]$ can be obtained by the right uncles of $S[3]$ (red and yellow nodes).
In order to have $\LF_{S[4 \dd 13]}$, the last right uncle (in yellow) needs to be trimmed. By Observation~\ref{obs:prefLF}, we take the shortest period of the trimmed string $S[9 \dd 13]=\texttt{bcdbc}$ and obtain the Lyndon factorization $\texttt{bcd},\texttt{bc}$. Thus we obtain the Lyndon factorization of $S[4 \dd 13]$, also depicted in~\cref{fig:Ltree}, which is $\texttt{c},\texttt{c},\texttt{bce},\texttt{bcd},\texttt{bc}$.
\end{example}
\begin{figure}[htpb]
\begin{center}
\includegraphics[width=0.8\textwidth]{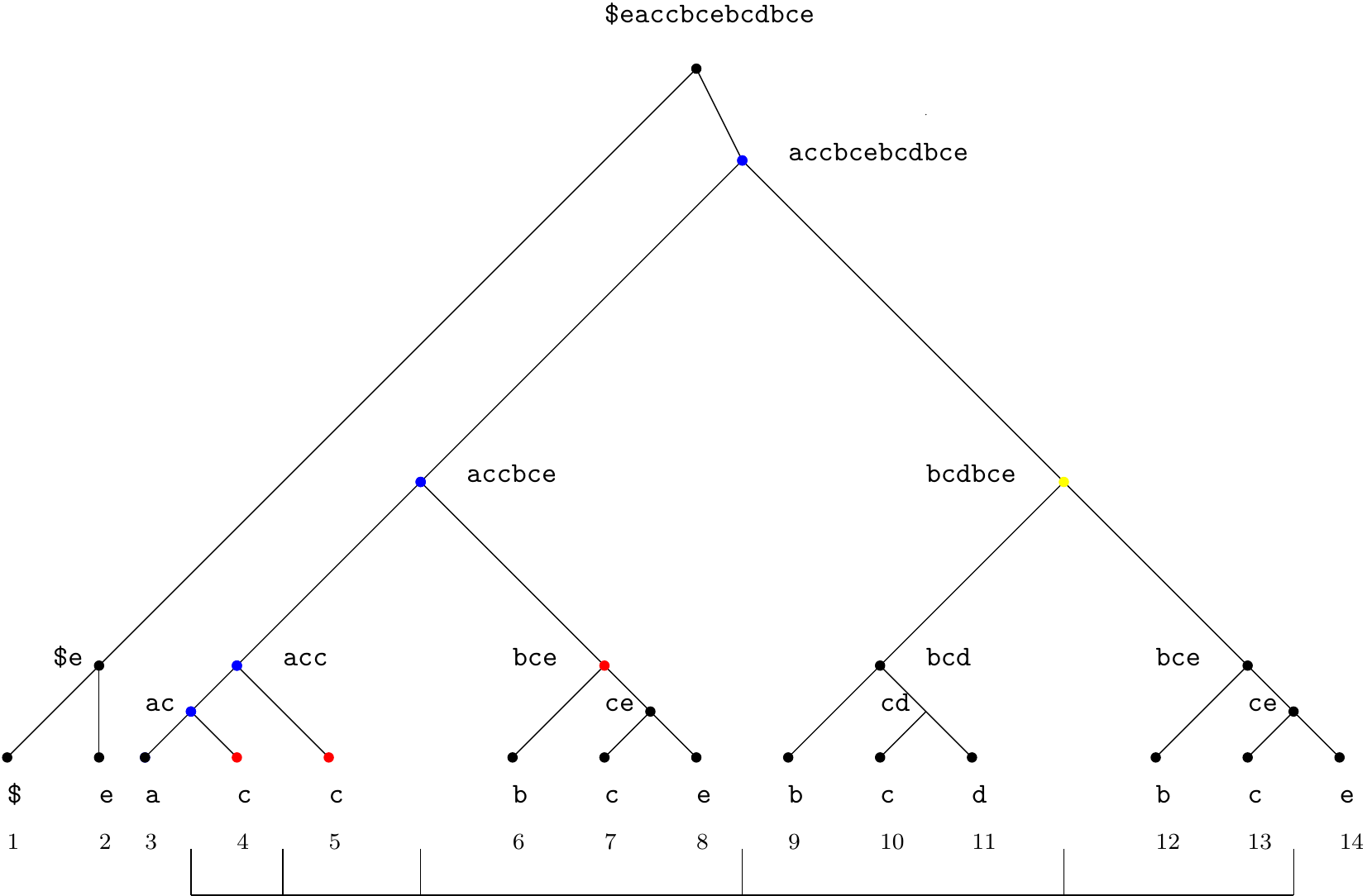}
\end{center}
\caption{Lyndon tree of string $S=\texttt{\$eaccbcebcdbce}$.}
\label{fig:Ltree}
\end{figure}

In order to make this computation efficient, heavy paths of $\LTree(S)$ are stored. Each non-leaf node $v$ stores, as $\rc(v)$, the length of its right child provided that it is \emph{not} present on its heavy path, or $-1$ otherwise. A balanced BST with all the nodes in the heavy path with positive values of $\rc$ given in bottom-up order is stored for every heavy path. It can be augmented in a standard manner to support the following types of queries on any subpath (i.e., ``substring'') of the path in $\Oh(\log n)$ time:
\begin{itemize}
  \item \textsf{longest-subpath}: the maximal value $\rc$ on the subpath
  \item \textsf{count-subpath}: the number of nodes with a positive value of $\rc$ on the subpath
  \item \textsf{select-subpath}$(i)$: the $i$th node with a positive value of $\rc$ on the subpath.
\end{itemize}
These precomputations take $\Oh(n)$ time and space. Finally, a lowest common ancestor data structure can be constructed over the Lyndon tree in $\cO(n)$ time and $\cO(n)$ space~\cite{Bender2000} supporting LCA-queries in $\cO(1)$ time per query. 

$\RLF_{S[i \dd j]}$ is represented as follows. First, a pair of nodes $(v_1,w')$ where $w'$ is the left child of $w$ is stored. Second, the Lyndon factorization of the right uncle $u$ obtained by recursively applying Lemma~\ref{lem:prefLF} is stored in a compact form in $\Oh(\log n)$ space. This representation can be computed in $\Oh(\log^2 n)$ time as follows:
\begin{itemize}
  \item The node $w$ being the LCA of $v_1$ and $v_2$ is computed in $\Oh(1)$ time. Then $w'$ is the left child of $w$.
  \item Each step of the recursive factorization of a pre-Lyndon string of Lemma~\ref{lem:prefLF} can be performed in $\Oh(\log n)$ time if a data structure for internal Period Queries of Kociumaka et al.~\cite{DBLP:conf/soda/KociumakaRRW15} is used. The total number of steps is $\Oh(\log n)$ since each step at least halves the length of the string in question.
\end{itemize}
Finally let us check that $\RLF_{S[i \dd j]}$ supports the desired types of queries in $\Oh(\log^2 n)$ time:
\begin{itemize}
  \item \textsf{longest}: Divide the path from $v_1$ to $w'$ into $\Oh(\log n)$ subpaths of heavy paths and for each of them ask a \textsf{longest-subpath} query. This takes $\Oh(\log^2 n)$ time.
  Compare the maximum of the results with the maximum length of a factor in the second part of the representation, in $\Oh(\log n)$ time.
  \item \textsf{count} is implemented analogously using \textsf{count-subpath}.
  \item \textsf{select}$(i)$: Use \textsf{count-subpath} queries to locate the subpath that contains the $i$-th factor in the whole factorization or check that none of the subpaths contains it. In the first case, use a \textsf{select-subpath} query. In the second case, locate the right factor in the second part of the representation in $\Oh(\log n)$ time.
\end{itemize}
Clearly the answers to the \textsf{longest} and \textsf{count} queries can be determined once upon the computation of $\RLF_{S[i \dd j]}$ and stored to support $\Oh(1)$-time queries. 

\subsection{Cross-Substring Queries}
In the case of this problem, we aim at computing the Lyndon factorization of a $k$-substring from the Lyndon factorizations of the $k$ subsequent substrings. To this end, the following characterization of Lyndon factorizations of concatenations of strings can be applied.

\begin{lemma}[see Lemma~6 and~7 in \cite{DBLP:conf/cpm/UrabeNIBT18}; originally due to \cite{DBLP:journals/mst/ApostolicoC95,DBLP:journals/tcs/DaykinIS94,DBLP:journals/tcs/INIBT16}]\label{lem:mergeLF}~\\ Assume that $\LF_U=(L_1)^{p_1},\ldots,(L_m)^{p_m}$ and $\LF_V=(L'_1)^{p'_1},\ldots,(L'_{m'})^{p'_{m'}}$. Then:
\begin{enumerate}[(a)]
  \item $\LF_{UV}=(L_1)^{p_1},\ldots,(L_c)^{p_c},Z^r,(L'_{c'})^{p'_{c'}},\ldots,(L'_{m'})^{p'_{m'}}$ for some $1 \le c \le m$, $1 \le c' \le m'$, string $Z$, and positive integer $r$.
  \item If $\LF_U$ and $\LF_V$ have been computed, then $c$, $c'$, $Z$, and $r$ from point (a) can be computed by $\Oh(\log|\LF_U|+\log|\LF_V|)$ lexicographic comparisons of strings, each of which is a 2-substring of $UV$ composed of concatenations of a number of consecutive factors in $\LF_U$ or $\LF_V$.
\end{enumerate}
\end{lemma}

In the case of a $k$-substring $S'$ of $S$, $\RLF_{S'}$ is a sequence of elements of two types: powers of Lyndon substrings of $S'$ and pairs of nodes $(v,w)$ that denote the bottom-up list of right uncles of nodes on the path from $v$ to $w$ in $\LTree(S)$. The length of $\RLF_{S'}$ is $\Oh(k \log n)$ and all its elements are stored in a left-to-right order in a balanced BST. For each element the maximum length of a factor and the number of factors are also stored in the BST. The BST is augmented with the counts of Lyndon factors so that one can identify in logarithmic time the element of the representation that contains the $i$th Lyndon factor in the whole factorization. This lets us implement the operation \textsf{select}$(i)$ on $\RLF_{S'}$ in $\Oh(\log (k \log n) + \log^2 n) = \Oh(\log^2 n)$ time (for $k \le n$) by first identifying the right element of the representation and then selecting the Lyndon factor in this element. The \textsf{longest} and \textsf{count} operations are performed in $\Oh(1)$ time if their results are stored together with the representation.

To compute $\RLF_{S'}$ for $F_1 \dots F_k$, we compute $\RLF_{F_i}$ for all $i=1,\ldots,k$ using internal queries from the previous subsection and then repetitively apply Lemma~\ref{lem:mergeLF}. Let $S''=F_1 \ldots F_{i-1}$ and assume that $\RLF_{S''}$ has already been computed. Then, by the lemma, $\RLF_{S''F_i}$ can be obtained by removing a number of trailing elements in $\RLF_{S''}$, a number of starting elements of $\RLF_{F_i}$, and merging them with at most one power of a Lyndon substring $Z$ of $S'$ in between.

\textit{Complexity.} Overall, the internal queries require $\Oh(k \log^2 n)$ time after $\Oh(n)$ time and space preprocessing. Every application of Lemma~\ref{lem:mergeLF} requires $\Oh(\log(k \log n))=\Oh(\log k + \log \log n)$ lexicographic comparisons of 2-substrings of $S'$, which gives $\Oh(\log k)$ is $k$ is polynomial in $n$. The 2-substrings can be identified in $\Oh(\log^2 n)$ time by a \textsf{select} query and then compared using the data structure of Lemma~\ref{lem:lce_rand}. This lemma requires $\Oh(n)$ space and answers $m$ queries on a $k$-substring in $\Oh((k+m)\log n)$ time, which gives $\Oh(k \log k \log n)=\Oh(k \log^2 n)$ time over all applications of Lemma~\ref{lem:mergeLF}.  In total, $\Oh(k \log^3 n)$ time is required to compute $\RLF_{S'}$.

\subsection{Round-Up}
The results of the two subsections can be combined into an algorithm for $k$\textsc{-Substring Lyndon Factorization}.

\begin{lemma}
  $k$\textsc{-Substring Lyndon Factorization} queries for $k \le n$ and $k$ polynomial in $n$ can be answered in $\Oh(k \log^3 n)$ time after $\Oh(n)$-time and space preprocessing.
\end{lemma}

Using the general scheme, we obtain a solution to the fully dynamic case.

\begin{theorem}
  Fully dynamic queries for the longest Lyndon substring, the size of the Lyndon factorization, or the $i$th element of the Lyndon factorization on a string of length up to $n$ can be answered in $\cO(\sqrt{n}\log^{1.5} n)$ time per query, using $\cO(n)$ space, after $\cO(n)$-time preprocessing.
\end{theorem}

\section{Final Remarks}

In this paper we presented techniques to obtain fully dynamic algorithms for several classical problems on strings; namely, for computing the longest common substring of two strings, the longest repeat, palindrome substring and Lyndon substring of a string.
We anticipate that the techniques presented in this paper are applicable in a wider range of problems on strings. 

We further anticipate that this new line of research will inspire more work on the lower bound side of dynamic problems. The currently known hardness (and conditional hardness) results for dynamic problems on strings have been established for dynamic pattern matching~\cite{DBLP:conf/stacs/CliffordGLS18,DBLP:conf/soda/GawrychowskiKKL18}. It would be interesting to investigate (conditional) lower bounds for the set of dynamic problems considered in this paper. 



\bibliographystyle{plain}
\bibliography{references}

\newpage

\appendix

\section{Proof of Lemma~\ref{lem:merge_ranges}}\label{app:ami}
Let us start with basic definitions on the suffix array and the ranges of substrings in the suffix array.

The \textit{suffix array} of a non-empty string $S$ of length $n$, denoted by $\SA(S)$, is an integer array of size $n+1$ storing the starting positions of all (lexicographically) sorted suffixes of $S$, i.e.~for all 
$1 < r \le n+1$ we have $S[\SA(S)[r-1] \dd n] < S[\SA(S)[r] \dd n]$. Note that we explicitly add the empty suffix to the array. The suffix array $\SA(S)$ corresponds to a pre-order traversal of all terminal nodes of the suffix tree $\mathcal{T}(S)$. We define a \emph{generalized suffix array} of $S_1,\ldots,S_k$, denoted $\SA(S_1,\ldots,S_k)$, as $\SA(S_1\#_1 \ldots S_k\#_k)$.

Let $S$ be a string of length $n$. Given a substring $U$ of $S$, we denote by $\range_S(U)$ the range in the $\SA(S)$ that represents the suffixes of $S$ that have $U$ as a prefix. Every node $u$ in the suffix tree $\tr(S)$ corresponds to an $\SA$ range $\range_S(\mathcal{L}(u))$. 

\begin{proof}[Proof of Lemma~\ref{lem:merge_ranges}]
  Let $X$ be a string of length $n$. We can precompute $\range_X(\mathcal{L}(u))$ for all explicit nodes $u$ in $\tr(X)$ in $\cO(n)$ time while performing a depth-first traversal of the tree. We use the following result by Fischer et al.~\cite{fischer_et_al:LIPIcs:2016:6066}.

\begin{claim}[\cite{fischer_et_al:LIPIcs:2016:6066}]
Let $U$ and $V$ be two substrings of $X$ and assume that $\SA(X)$ is known. Given $\range_X(U)$ and $\range_X(V)$, $\range_X(UV)$ can be computed in time $\cO(\log \log n)$ after $\cO(n \log \log n)$-time and $\cO(n)$-space preprocessing.
\end{claim}

We use the data structure of the claim for the generalized suffix array $\SA(S,T)$. The range of a substring $U$ is denoted as $\range_{S,T}(U)$. We assume that each element of $\SA(S,T)$ stores a 1 if and only if it origins in $T$ and prefix sums of such values are stored. This lets us check if a given range of $\SA(S,T)$ contains any suffix of $T$ in $\Oh(1)$ time. We also use the GST $\mathcal{T}(S,T)$.

By Corollary~\ref{cor:LAQ}, the loci of $U$ and $V$ in $\mathcal{T}(S,T)$ can be computed in $\Oh(\log \log n)$ time. This lets us recover the ranges $\range_{S,T}(U)$ and $\range_{S,T}(V)$. By the claim, we can compute $\range_{S,T}(UV)$ in $\Oh(\log \log n)$ time. Then we check if $UV$ is a substring of $T$ by checking if the resulting range contains a suffix of $T$; as already mentioned, this can be done in $\Oh(1)$ time. The data structures can be constructed in $\Oh(n \log \log n)$ time and use $\Oh(n)$ space. This concludes the proof of the first part of the lemma.

As for the second part, it suffices to apply binary search over $U$ to find the longest prefix $U'$ of $U$ that is a substring of $T$. If $U'=U$, we apply binary search to find the longest prefix $V'$ of $V$ such that $UV'$ is a substring of $T$. Binary searches result in additional $\log n$-substring in the query complexity. The approach for computing the longest suffix is analogous.
\end{proof}

\section{Reducing the Number of Dimensions for Extended HIA Queries}\label{app:dimensions}
The proof of Lemma~\ref{lem:extended-HIA} used RMQ on a collection of points in a 6-dimensional grid. This is a significant number of dimensions. However, it can be reduced to 4 by applying the methods that were used in \cite{Amir2017}. This change does not improve the complexities in the $\Ohtilde$-notation.

Let $\pi$ be a heavy path and $u$ be its topmost node. Assume that $u$ contains $m$ leaves in its subtree. We then denote by $L(\pi)$ the \emph{level} of the path $\pi$, which is equal to $\lfloor \log m \rfloor$. We use the following property of heavy-path decompositions.

\begin{observation}\label{obs:heavy_levels}
  For any leaf $v$ of $\mathcal{T}$, the levels of all heavy paths visited on the path from $v$ to the root of $\mathcal{T}$ are distinct.
\end{observation}

Let us rearrange the children in $\mathcal{T}_1$ and $\mathcal{T}_2$ so that the heavy edge always leads to the leftmost child.

Instead of just four collections of points, we create $\Oh(\log^2 n)$ such collections. Instead of $\mathcal{P}^{I}$, we create collections $\mathcal{P}^{I}_{x,y}$ for all $0 \le x,y \le \lfloor \log n \rfloor$. For $i \in \{1,\ldots,n\}$ and $j=1,2$, let $m_j(i)$ be the position of the leaf $i$ in a left-to-right traversal of all leaves in $\mathcal{T}_j$. For each leaf $i$, if $\pi_1$ and $\pi_2$ are any heavy paths on the path from leaf $i$ to the root of $\mathcal{T}_1$ and $\mathcal{T}_2$, respectively, $d_j=d(\pi_j,i)$ for $j=1,2$, $p_1=L(\pi_1)$ and $p_2=L(\pi_2)$, then we insert the point
\begin{itemize}
  \item $(m_1(i),m_2(i),i-d_1,i+d_2)$ to $\mathcal{P}^I_{p_1,p_2}$ with weight $d_1+d_2$.
\end{itemize}
A similar transformation can be performed on $\mathcal{P}^{II}$, $\mathcal{P}^{III}$, and $\mathcal{P}^{IV}$.

Assume we are to answer an extended HIA query for $v_1$, $v_2$, $a$, and $b$. As before, we consider all the prefix fragments $\pi'_j \in H(v_j,\mathcal{T}_j)$, $j=1,2$. Let $\pi'_j$ connect node $x_j$ with its descendant $y_j$ and let $e_j$ and $f_j$ be the minimal and maximal $m_j$ number of a leaf in the subtree of $x_j$ excluding the descendants of $y_j$. If $\pi'_j$ is a prefix fragment of a heavy path $\pi_j$, we set $p_j = L(\pi_j)$. Then for this pair of prefix fragments we ask the following query:
\begin{itemize}
  \item $\RMQ_{P^I_{p_1,p_2}}([e_1,f_1],[e_2,f_2],[a,\infty),(-\infty,b+1])$.
\end{itemize}
We treat the remaining cases similarly. On the first two components we take either the intervals as shown in the formula above or the interval corresponding to the subtree of $y_j$.

Thus the number of dimensions reduces to 4. Moreover, it can be further reduced to 3 if we apply the extended HIA query to solve the \textsc{Three Substrings LCS} problem for $W$ being a prefix or a suffix of $T$, as it is the case in Section~\ref{sec:1-1.2}.
\end{document}